\documentclass[final]{article}    
\usepackage[utf8]{inputenc}
\usepackage[T1]{fontenc}
\usepackage[english]{babel}
\usepackage{multirow}
\usepackage{booktabs}
\usepackage{amssymb}
\usepackage{amsmath}
\usepackage{graphicx}
\usepackage{amsthm}
\usepackage{float}

\usepackage{url}

\usepackage{hyperref}
\hypersetup{letterpaper, urlcolor=black}

\usepackage{xcolor}

\usepackage{amsmath}
\usepackage[vlined]{algorithm2e}
\SetKw{KwAnd}{and}
\SetVlineSkip{0pt}  
\usepackage{tikz}
\usetikzlibrary{matrix}

\DeclareMathAlphabet\mathbfcal{OMS}{cmsy}{b}{n}

\usepackage{authblk}

\newcommand{\summenproblem}{\textsc{K-Smallest Subsets}}
\newcommand{\NP}[0]{$\mathcal{NP}$}
\newcommand{\Pclass}[0]{$\mathcal{P}$}
\newcommand{\NPcOurResult}[1]{\hyperref[#1]{$\mathbfcal{NP}$\textbf{-c}}}
\newcommand{\POurResult}[1]{\hyperref[#1]{$\mathbfcal{P}$}}
\newcommand{\Wzwei}[0]{$W[2]$}

\newcommand{\bigO}[1]{\mathcal{O}(#1)}
\newcommand{\plusplus}[0]{\raisebox{0.3ex}{\tiny\textbf{++}}}

\DeclareMathOperator{\costs}{costs}

\newcommand{\Cany}{C_\textsc{any}}
\newcommand{\Cequal}{C_\textsc{equal}}
\newcommand{\Clevel}{C_\textsc{level}}
\newcommand{\Cflip}{C_\textsc{flip}}
\newcommand{\Cdist}{C_\textsc{dist}}

\newcommand{\summe}{\psi}
\newcommand{\SpecialPartition}{\textsc{Partition}}

\newtheorem{theorem}{Theorem}

\newtheorem{corollary}[theorem]{Corollary}

\newtheorem{lemma}[theorem]{Lemma}

\newtheorem{example}{Example}
\newtheorem*{remark}{Remark}

\title{\mytitle}

\begin{document}

\title{On the Hardness of Bribery Variants in Voting with CP-Nets\footnote{The final publication is available at Springer via \url{http://dx.doi.org/10.1007/s10472-015-9469-3}.}}

\author[1]{Britta Dorn}
\author[2]{Dominikus Krüger}

\affil[1]{Faculty of Science, Dept.~of Computer Science,
			University of Tübingen,
			Germany, 
			\texttt{britta.dorn@uni-tuebingen.de}}

\affil[2]{Dominikus Krüger,
			Institute of Theoretical Computer Science,
			Ulm University,
			Germany,
			\texttt{dominikus.krueger@uni-ulm.de}}

\date{}

\maketitle

\begin{abstract}
We continue previous work by Mattei \textit{et al.}~\cite{mattei2013bribery}
in which they study the computational complexity of bribery schemes when voters have conditional preferences modeled as CP-nets. For most of the cases they considered, they showed the bribery problem is solvable in polynomial time. Some cases remained open---we solve several of them and extend the previous results to the
case that voters are weighted. Additionally, we consider negative (weighted) bribery in CP-nets, when the briber is not
allowed to pay voters to vote for his preferred candidate.
\end{abstract}

\section{Introduction}
This work is based on the findings of Mattei, Pini, Rossi, and Venable~\cite{mattei2013bribery}. 
They considered the computational complexity of computing optimal bribery schemes in voting scenarios in which the voters decide on a fixed assignment for a common set of issues for which they might have conditional preferences. The typical example for this setting is the choice of a common meal consisting of several courses and drinks. Here it is a natural assumption that some voters' preference over their choice of wine is conditioned on the kind of meat being served, which again might depend on the choice of previously chosen dishes. In such a setting, the set of candidates to vote on is exponentially large in the number of common issues. A compact and convenient way to condense and represent the voters' possibly conditional preferences over this set is given by the CP-net formalism introduced by Boutilier~\textit{et al.}~\cite{boutilier2004cp}, see also earlier works by Boutilier~\textit{et al.}~\cite{boutilier1997constraint, boutilier1999reasoning}.  \\

In a CP-net, a voter's dependencies for a set of issues are modelled by a directed graph, and the conditional aspect of his preferences is expressed by so-called {\it ceteris paribus} conditional preference statements for each issue. 

Preference aggregation in CP-nets was first investigated by Rossi~\textit{et al.}~\cite{rossi2004mcp} and was further pursued in various works such as Purrington and Durfee~\cite{purrington2007making}, Li~\textit{et al.}~\cite{li2010efficient, li2011majority}, Xia~\textit{et al.}~\cite{xia2008voting}, Conitzer~\textit{et al.}~\cite{conitzer2011hypercubewise}, to name just a few. Among the approaches for voting with CP-nets, there are one-step methods that, given the voters' CP-nets, derive a linear preference order (or relevant parts of it) over the candidates and apply one of the existent rules for elections where preferences are given as linear orders. Lang~\cite{lang2007vote} proposed a sequential approach which aggregates preferences issue by issue, see also the work of Lang and Xia~\cite{lang2009sequential} and Xia~\textit{et al.}~\cite{xia2007sequential, xia2007strongly}. 
In their work, Mattei~\textit{et al.}~\cite{mattei2013bribery} use the one-step methods One-Step-Plurality ($OP$), One-Step-Veto ($OV$), and One-Step-$k$-Approval ($OK$) (and a variant called $OK^*$) which are derived from the well-known voting rules for linear preference orders, as well as the sequential voting system Sequential Majority ($SM$).\\ 

A central challenge in the field of computational social choice consists in determining the computational complexity of voting problems~\cite{brandt2013comsoc, faliszewski2009richer}.
One seeks to assess the extent to which voting systems are---in some sense---vulnerable to or resistant against manipulative actions such as strategic behavior (manipulation), election control, 
or bribery. These properties are measured in terms of computational hardness of the corresponding decision problems. For the definition of these problems and for surveys on this topic we refer to the article by Faliszewski~\textit{et al.}~\cite{faliszewski2009richer} and the bookchapter by Faliszewski~\textit{et al.}~\cite{faliszewski2010protect} on manipulation, bribery, and control, the article by Faliszewski and Procaccia~\cite{faliszewski2010war} on manipulation, the survey by Rothe and Schend~\cite{rothe2013shields} on manipulation and control, and to the bookchapters by Conitzer and Walsh~\cite{conitzer2015handbook} for the manipulation problem, and by Falizewski and Rothe~\cite{faliszewski2015handbook} for bribery and control.

In this work, we consider the {\it bribery} problem which asks whether an external agent, the briber, can influence the voters by spending money on changing their preferences over the candidates in such way that his preferred candidate wins, respecting a given budget. This problem was introduced and intensively studied by Faliszewski~\textit{et al.}~\cite{faliszewski2009hard}. The original version deals with individual but fixed prices for each voter, independent of the specific changes made if the voter is bribed. Several variants have been considered~\cite{faliszewski2009hard}, among them the model of nonuniform bribery introduced by Faliszewski~\cite{faliszewski2008nonuniform} and the model of microbribery studied by Faliszewski~\textit{et al.}~\cite{FHHR09} which may incorporate the amount of change the briber asks for. The {\sc Swap Bribery} problem introduced by Elkind~\textit{et al.}~\cite{elkind2009swap} additionally takes into account  the ranking aspect of the voters' preferences by assigning costs for swapping two consecutive candidates in a preference order.

Given that bribery of an election represents some kind of dishonest behavior, one is interested in deriving hardness results for voting systems with respect to the computational complexity of the bribery problem, even if this only constitutes a worst case analysis. 

However, the bribery problem can also be considered in a positive way in terms of convincing voters by (possibly) cost-involving actions like campaigning to change their votes, and is referred to as {\it election campaign management} then, see the works of Elkind~\textit{et al.}~\cite{elkind2009swap}, Elkind and Faliszewski~\cite{elkind2010approximation}, Schlotter~\textit{et al.}~\cite{schlotter2011campaign}, and Baumeister~\textit{et al.}~\cite{baumeister2012lazy}.\\
Considering the bribery problem as a campaign management problem, one is of course less interested in obtaining computational hardness than in finding efficient polynomial time algorithms.

In this work, we further investigate the complexity of the bribery problem in the setting of voting with CP-nets as initiated by Mattei~\textit{et al.}~\cite{mattei2013bribery}. The difference to the `classical' setting where preferences are given as linear orders is that the briber does not directly execute changes in the preference orders of the voters, but in the CP-nets, i.e., he can affect the dependencies. From the point of view of campaign management, this seems very natural as well. One can easily imagine a systematic campaign to convince voters to drop dependencies in their preferences, without having to deal with their implicit preference order. For instance, if a voter prefers some special kind of wine in the case that meat is served as a main dish, one might run a campaign for the quality of the alternative varieties of wine offered and therewith directly affect the CP-net, without having to care about the induced preference order.\\
Bribery in CP-nets has also been investigated by Maran~\textit{et al.}~\cite{maran2013framework} in the context of interaction and influence among voters and by Pini~\textit{et al.}~\cite{pini2013bribery} in connection to representation of the voters' preferences via soft constraints.\\

In the `classical' setting of unconditional preferences, the basic bribery problem is solvable in polynomial time for the voting rule $k$-approval~\cite{faliszewski2009hard}, and therefore also for its special cases plurality and veto. In contrast, the {\sc Swap Bribery} problem for $k$-approval is solvable in polynomial time if all swaps have a cost of one, but becomes \NP-complete as soon as different costs are allowed, already for the case $k=2$ (see the work of Dorn and Schlotter~\cite{dorn2012multivariate}, Betzler and Dorn~\cite{BD10} and Elkind~\textit{et al.}~\cite{elkind2009swap}); the complexity hence depends---additionally to the voting system used---on the amount of change that the briber has to pay for. 

For bribery in the CP-net setting, the amount of change that the briber has to pay for is incorporated by a cost scheme. Mattei~\textit{et al.} introduce five cost schemes, called $\Cequal$, $\Cflip$, $\Clevel$, $\Cdist$, and $\Cany$, which are inspired by the classical setting, ranging from unitary costs irrespective the amount of change per voter, over the scenario of swap bribery where specific costs for swaps in the preferences can be modelled, to the case where arbitrary cost can be incorporated.
Additionally, these cost schemes are extended by a cost vector ${\bf Q} \in (\mathbb{N})^n$ which allows for modelling an individual cost factor for each voter. Also, different bribery actions are considered. Mattei~\textit{et al.} distinguish the cases that the briber can affect changes in so-called independent variables only (IV), meaning that he can only make changes concerning issues for which the preferences are independent of the outcome for other issues, in dependent variables only (DV), or in all kinds of variables (IV+DV).

In most of the cases they considered, Mattei~\textit{et al.} obtained that the bribery problem is easy, i.e., solvable in polynomial time. 
An overview of their results is given in Table~\ref{table:results-mattei}. Three cases with different variants remained unanswered: 
The complexity for the voting systems $OP$, $OV$, $OK^*$ with the cost scheme $\Clevel$ for bribery actions $DV$ and $IV+DV$, the complexity for the voting systems $OP$, $OV$, $OK^*$ with the cost scheme $\Cany$ for all bribery actions, and the complexity for $SM$ with the cost scheme $\Cdist$ for all bribery actions.
The first two cases were believed to be \NP-complete. We show that these two are solvable in polynomial time as well. 

Moreover, we extend Mattei~\textit{et al.}'s results  for the case that voters are weighted. So far, the weighted version of the bribery problem in CP-nets was considered for the voting system $SM$ only. Mattei~\textit{et al.} obtained \NP-completeness for all cost schemes and all bribery actions in this case. Additionally, for the cost scheme $\Cflip$, they showed that the problem is solvable in polynomial time if no individual voter costs are allowed, see Table~\ref{table:results-mattei}.
We complement their results by proving \NP-completeness for most variants of the weighted version for the other voting rules as well, and further illuminate the influence of individual voter costs for the complexity of the weighted versions. We find that considering individual extra costs for voters only makes a difference in terms of complexity for the voting system weighted $SM$ combined with the cost schemes~$\Cflip$ and $\Clevel$. \\
Additionally, we investigate the case of negative bribery, a variant introduced by Faliszewski \textit{et al.}~\cite{faliszewski2006complexity}, where it is not allowed to pay voters to vote for the preferred candidate, both for the unweighted and weighted versions. We obtain that the unweighted negative bribery problem is likewise solvable in polynomial time for almost all variants considered so far, except for $SM$ combined with $\Cequal$ where it is \NP-complete, and for $SM$ combined with $\Cdist$, which still is unsolved, like its correspondent in the positive version. The weighted negative version is \NP-complete for all variants considered so far.
Our results are summarized in Table~\ref{table:results} in Section~\ref{sec:results}.\\

\noindent
This work is organized as follows. 
In Section~\ref{sec:prel}, we give the basic notions used from CP-nets, voting, and bribery. We introduce the setting of Mattei~\textit{et al.}, the definitions of the bribery variants to be investigated, and the \NP-complete problems that we use in our reductions. 
Our results are contained in Section~\ref{sec:complexity}. We start with proving Theorem~\ref{thm:sums} for finding the `cheapest subsets' of a finite set which helps in solving some of the open problems from Mattei~\textit{et al.}~\cite{mattei2013bribery} in the unweighted variants of the bribery problem. The subsequent subsections deal with the complexity of the weighted case, where reductions from the \NP-complete {\sc Partition} problem are given, and investigations for the unweighted and weighted cases of negative bribery. An overview of our results and open problems is provided in Section~\ref{sec:results}, together with a discussion and directions for future research.

\begin{table}[htb]
\caption{Complexity results obtained by Mattei~\textit{et al.}~\cite{mattei2013bribery} for the bribery problem in CP-nets. \Pclass{} stands for solvability in polynomial time, \NP-c for \NP-completeness. The given results hold for the bribery actions $IV$, $DV$, and $IV+DV$. Questions that were unsolved by Mattei {\it et al.} are marked with `?'.
The case of weighted $SM$ with $\Cflip$ labeled with two complexity classes means that the problem can be solved in polynomial time if no individual voter costs are allowed, and that it is \NP-complete if individual costs are taken into account. \label{table:results-mattei}}

{\footnotesize
\centering
\begin{tabular}{lccccc}
\toprule 
 & $\Cequal$ & $\Cflip$ & $\Clevel$ & $\Cany$ & $\Cdist$  \\
\midrule
 SM & \NP-c & \Pclass & \Pclass & \Pclass & ?  \\
  OP & \Pclass & \Pclass & \Pclass{} for $IV$/  ? for $DV$, $IV+DV$ & ?  & \Pclass \\

  OV & \Pclass & \Pclass & \Pclass{} for $IV$/  ? for $DV$, $IV+DV$  & ?  & \Pclass \\
  OK* & \Pclass & \Pclass & \Pclass{} for $IV$/   ? for $DV$, $IV+DV$  & ? & \Pclass \\
\midrule
  weighted~SM & \NP-c & \Pclass{}, \NP-c & \NP-c & \NP-c & \NP-c \\
\bottomrule

\end{tabular}}
\end{table}

\section{Preliminaries}\label{sec:prel}

We mostly use notation and definitions as introduced by Mattei \textit{et al.}~\cite{mattei2013bribery}. 

\subsection{CP-nets} 

In our setting, we are given a set of $m$ issues~$M=\{X_1, \dots, X_m\}$ each of which has a binary domain $D(X_j) = \{x_j, \overline{x}_j\}$, where $j \in \{1, \dots, m\}$. A complete assignment to all issues is called an {\it outcome} or a {\it candidate}, the set of candidates $D(X_1) \times \dots \times D(X_m)$ hence consists of $2^m$ elements. Each of the $n$~voters has (possibly) conditional preferences over the values assigned to the issues: the value assigned to an issue by a voter might affect the values of other issues. In this case, the latter issues are called {\it dependent} issues; an issue is called {\it independent} if its value does not depend on the value of another issue. The dependencies can be modeled by a directed \emph{dependency graph} having vertex set~$M$ and a directed edge going from~$X_i$ to~$X_j$ if and only if the assignment of the value of~$X_j$ depends on the assignment of the value of~$X_i$. 
A {\it CP-net} over $M$ consists of a dependency graph together with a {\it conditional preference table} for each issue $X_i$ (or vertex of the dependency graph, respectively) where the voter
specifies a strict total order over the values of $X_i$ for each complete assignment of the issues on which it depends on. Each of these specifications is referred to as a {\it cp-statement}. 
For example, if for an independent issue $A$, the assignment~$A=a$ is unconditionally preferred to $A=\overline{a}$, the cp-statement is written as $a>\overline{a}$; if the assignment~$A=\overline{a}$ is preferred to $A=a$, we write $\overline{a}>a$.  For a set of issues~$\{A,B\}$ with domains $D(A)=\{a,\overline{a}\}$ and $D(B) = \{b,\overline{b}\}$ in which~$B$ is dependent on~$A$, if $a$ is unconditionally preferred over $\overline{a}$, and $\overline{b}$ is preferred over $b$ in case the value of $a$ is assigned to $A$, and $b$ is preferred over $\overline{b}$ in case the value of $\overline{a}$ is assigned to $A$, we have one cp-statement for $A$, namely $a>\overline{a}$, and two cp-statements for $B$, which we write as  
\begin{align*}
a: \overline{b}>b, \\
\overline{a}: b >\overline{b}.
\end{align*}

A CP-net is called \emph{acyclic} if the corresponding dependency graph is acyclic; it is called \emph{compact} if the number of issues each issue depends on is bounded by a constant. We remark that the original definition of CP-nets by Boutilier \textit{et al.}~\cite{boutilier2004cp} does not require acyclicity, but this assumption is natural \cite{lang2007vote} and also employed by Mattei \textit{et al.}~\cite{mattei2013bribery}. In acyclic CP-nets, there is only one most preferred candidate which can then be found in linear time by going through the CP-net in topological order and assigning the most preferred value to each issue according to the cp-statements~\cite{boutilier2004cp}.\\
Throughout this work, we assume that the voters' preferences on the set of issues are given by acyclic and compact CP-nets. The collection of CP-nets of all voters is called a \emph{profile}.

\begin{example}\label{ex:friends}
A group of friends wants to go on vacation. They can choose between going to Austria or Italy for hiking or alpine skiing in the summer or the winter holidays. The three issues to decide on are hence \emph{Where}, \emph{When}, and \emph{What}. So the corresponding domains are $D(\text{Where})=\{\text{Austria},\text{Italy}\}$, $D(\text{When})=\{\text{summer},\text{winter}\}$, and $D(\text{What})=\{\text{skiing},\text{hiking}\}$. A possible candidate would be \emph{skiing} in \emph{summer} in \emph{Italy}. Figure~\ref{fig:example} shows two possible conditional preferences modelled as CP-nets, each of them consisting of a dependency graph and a conditional preference table.
\end{example} 

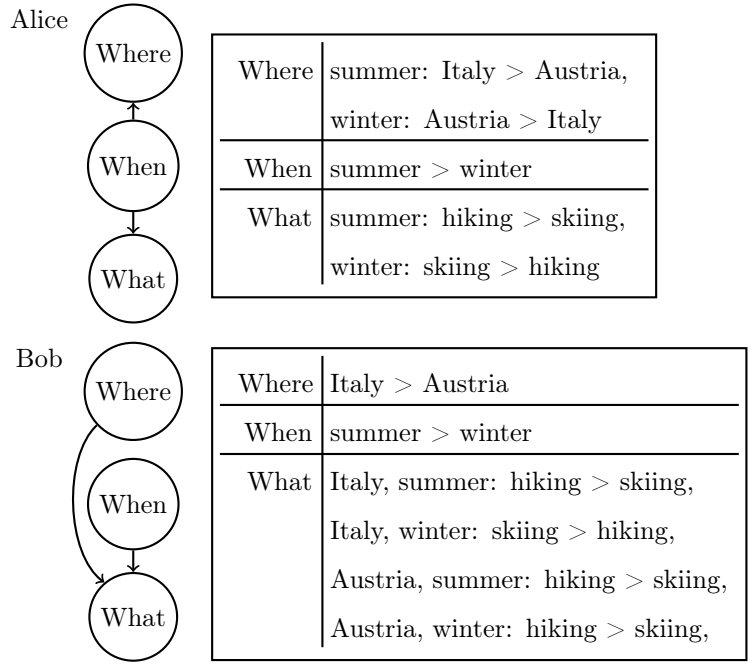
\begin{figure}[h!]
\centering
\begin{tikzpicture} [->,thick,align=center,xscale=2.5,yscale=1.5]  
  \path[every node/.style={draw,thick,circle, minimum width=0.6cm, inner sep=3pt}]
  	   (-0.5,5.3) node[draw=none] {Alice}
  	   (0,5)   node (a1)  {Where}
       (0,4)   node (a2) {When}
       (0,3)   node (a3) {What}
       (-0.5,2.3) node[draw=none] {Bob}
  	   (0,2)   node (c1)  {Where}
       (0,1)   node (c2) {When}
       (0,0)   node (c3) {What};
  \path[every node/.style={font=\sffamily\small}]
  	(a2)		edge node {} (a1)
			edge node {} (a3)
    	(c1)		edge[out=230,in=130,looseness = 0.6,->] node {} (c3)
    	(c2)		edge node {} (c3);
    	
  \matrix(alice)[draw,matrix of nodes,nodes in empty cells,
        nodes={align=left,text width=4.1cm,text height=0.35cm,},
        column 1/.style={nodes={text width=1.1cm,align=right}}
    ] at (1.6,4){
        \phantom{j}Where 	& summer: Italy > Austria, \\
         		& winter: Austria > Italy \\ 
        \phantom{j}When 	& summer > winter\phantom{j}\\
        What 	& summer: hiking > skiing,\\
        			& winter: skiing > hiking\\
    };
    \draw[-](alice-3-1.north west)--(alice-3-2.north east);
    \draw[-](alice-4-1.north west)--(alice-4-2.north east);
    \draw[-](alice-1-1.north east)--(alice-5-1.south east);
    
  \matrix(charlie)[draw,matrix of nodes,nodes in empty cells,
        nodes={align=left,text width=5.3cm,text height=0.35cm,},
        column 1/.style={nodes={text width=1.1cm,align=right}}
    ] at (1.84,1){
        \phantom{j}Where &  Italy > Austria \\ 
        \phantom{j}When 	& summer > winter\phantom{j}\\
        What 	& Italy, summer: hiking > skiing,\\
        			& Italy, winter: skiing > hiking,\\
        			& Austria, summer: hiking > skiing,\\
        			& Austria, winter:  hiking > skiing,\\
    };
    \draw[-](charlie-3-1.north west)--(charlie-3-2.north east);
    \draw[-](charlie-2-1.north west)--(charlie-2-2.north east);
    \draw[-](charlie-1-1.north east)--(charlie-6-1.south east);   
   
\end{tikzpicture}
\caption{Two CP-nets over the issues \emph{Where}, \emph{When}, and \emph{What}. The possible values are Austria or Italy in the summer or the winter holidays for hiking or alpine skiing. Alice thinks that a summer vacation always is the better choice, but if it has to be winter she prefers alpine skiing in Austria.
Bob thinks that alpine skiing is not an appropriate summer activity and that only Italy has good ski regions.}
\label{fig:example}
\end{figure}

CP-nets only define a partial order over the candidates, i.e., some candidates are incomparable. There are several ways to create strict total orders over the candidates (\cite{brandt2013comsoc,conitzer2012approximating}). We use the linearization method that is also used by Mattei~\textit{et al.}: Each voter provides a fixed strict total order $X_{i_1} > X_{i_2} > \dots > X_{i_m}$ over the issues $X_1, \dots, X_m$ such that each issue is independent from all issues following it in this order; this is possible because the CP-net is acyclic. Then we associate a binary vector of length~$m$ to each candidate, for which the entry at position~$j$ corresponds to issue $X_{i_j}$ in the voter's fixed ordering. This entry is set to~$0$ if the preferred value of its binary domain is assigned to it, and $1$ otherwise. Hence, the most preferred candidate is associated to the vector~$(0,\dots,0)$, and the least preferred one to the vector~$(1, \dots, 1)$.
Given a candidate, the next best candidate can be found efficiently by increasing the binary number represented by its vector by one.\\

In the general setting, every voter may have an individual fixed total order over the issues. The voting system $SM$ (defined below) relies on a sequential approach to voting and therefore requires existence of a strict total order $\mathcal{O}$ over the issues such that for each CP-net, each issue is independent from all issues following it in~$\mathcal{O}$. A profile that fulfills this property is called an $\mathcal{O}$-\emph{legal} profile by Lang~\cite{lang2007vote}, see also the work of Lang and Xia~\cite{lang2009sequential}. Mattei~\textit{et al.}~\cite{mattei2013bribery} use this notion as well as the notion of a \emph{constant linearization scheme across all agents} for such an order~$\mathcal{O}$. In Example~\ref{ex:friends}, the profile consisting of Alice's and Bob's CP-nets is $\mathcal{O}$-legal for the order~$\mathcal{O}$: {\it When} $>$ {\it Where} $>$ {\it What}.

\subsection{Voting} 
Given a profile of CP-nets over a set of issues, one can determine the winning outcome(s) by a voting rule which maps a profile to a set of candidates. 
In the \emph{unique winner} model, the rule determines a single winning candidate, whereas in the \emph{non-unique} or \emph{co-winner} model, the output of the rule is a whole set of candidates which are all considered as winners. In this case, the notion of a voting correspondence is also used for the notion of a voting rule. We consider the same voting rules for CP-nets as Mattei~\textit{et al.}~\cite{mattei2013bribery}.
\begin{itemize}
\item Sequential majority ($SM$): Given a total order $\mathcal{O}$ for which the profile is $\mathcal{O}$-legal, we follow this order issue by issue, and execute a majority vote for each issue. The voters fix the winning value of the corresponding issue in their CP-net and then go on to the next issue. The (unique) winning candidate is the combination of the winners of the individual steps taken. 
\item One-Step-$k$-Approval ($OK$): The $k$ most preferred candidates of each voter obtain $1$ point each, the remaining candidates obtain $0$ points. The (co-)winners of the election are the candidates with the maximum number of points. In particular, we consider the two following special cases: 
 \begin{itemize}
 \item One-Step-Plurality ($OP$), where $k=1$, i.e., only the most preferred candidate of each voter has to be considered, and 
 \item One-Step-Veto ($OV$), where $k=2^m-1$, i.e., only the least preferred candidate of each voter has to be provided. 
  \end{itemize}
\end{itemize}
For $OK$,  
$\mathcal{O}$-legality of the profile is not necessarily required. Mattei \textit{et al.} also consider the special case of~$OK$, denoted by $OK^*$, where $k$ is a power of two. 
For more details and several examples for the use of these rules, we refer to Mattei~\textit{et al.}~\cite{mattei2013bribery}.

 \subsection{Bribery}\label{subsection:Bribery}
 We consider the problem that an external agent, the  \emph{briber}, who knows the CP-nets of all voters, asks them to execute changes in their cp-statements. Mattei \textit{et al.}~\cite{mattei2013bribery} define this in a way that the briber can ask the voters to flip the value of one or more issues in their CP-nets, which might imply several further changes, according to the dependencies. They distinguish the case that the briber can ask for a change in the cp-statements of the independent issues only (IV), the dependent issues only (DV), or in any cp-statement (IV+DV).  
\\
Moreover, they introduce five cost schemes:
 \begin{itemize}
 \item $\Cequal$, where any amount of change in the CP-net has the same unit cost, 
 \item $\Cflip$, where the cost of changing a CP-net is the total number of individual cp-statements that must be flipped to obtain the desired change,
 \item $\Clevel$, where the cost to flip a cp-statement is linked to the depth of the associated issue in the dependency graph. \\ The cost of a bribery is computed as $\sum_{X_j \in M} \text{flip}(X_j)\cdot(k+1-\text{level}(X_j))$, where $k$ is the number of levels in the CP-net, $\text{level}(X_j)$ corresponds to the depth of issue $X_j$ in the dependency graph, and $\text{flip}(X_j)$ is the number of flips performed in cp-statements associated with $X_j$. (Note that for the voting rules we consider, we have $\text{flip}(X_j)\in \{0,1\}$ for all $j$, see the remark below.) More precisely, $$\text{level}(X_j)= \begin{cases} 1 & \text{if }  X_j \text{ is an independent issue;} \\ i+1 & \text{if level$(X_i) \leq i$ holds for all parents $X_i$ of $X_j$ and} \\ &  \text{there is at least one parent $X_l$ with level$(X_l) = i$}.\end{cases}$$
  \item $\Cany$, where the cost is the sum of the flips, each weighted by a specific cost, and
 \item $\Cdist$, where we require a fixed order of the issues for each voter (but not necessarily the same for each of them), which induces a strict total order over all candidates. The cost to bribe a voter to make candidate $c$ his top candidate is the number of candidates who are better ranked than $c$ in this order.
\end{itemize}

\noindent

\begin{remark}
 We remark that for the considered voting rules and for each reasonable bribery, $\text{flip}(X_j)$ in the definition of $\Clevel$ is equal to $0$ or to $1$ for each issue $X_j$. We call a bribery reasonable if there is no lower priced bribery that has the same benefit for the briber.
Clearly, $\text{flip}(X) \in \{1,0\}$ holds for an independent issue~$X$ because there is only one cp-statement associated with an independent issue. 
All dependent issues have more than one cp-statement, but for the four voting rules we consider, it is never reasonable to bribe for changes in more than just one cp-statement per issue, as those additional changes may generate additional costs but are of no help.
\begin{itemize}
\item For $SM$, when the majority vote is executed for an issue $X$, this means that for each voter all the issues that $X$ depends on already have a fixed value. Therefore only one cp-statement per issue is relevant. 
\item For $OP$, the briber only has to change the most preferred candidate of a voter in case he decides to bribe him. This is achieved by changing the value of some of the issues in the corresponding CP-net, and this can be done by flipping a single cp-statement in the issues that are concerned.
\item For $OV$, this is almost the same. The briber only has to change the least preferred candidate, which means that the value of some issues has to be flipped, which can be done by flipping a single cp-statement per issue. 
\item For $OK$, we have a set of approved candidates and a set of disapproved candidates. There is at most one cp-statement per issue which implies a change of these two sets if it is flipped. A flip of the other cp-statements only changes the order of the candidates within these sets. 
\end{itemize}
\end{remark}

\noindent
Additionally, these cost schemes are extended by a cost vector ${\bf Q} \in (\mathbb{N})^n$ which allows for modelling an individual cost factor for each voter. The factor for voter~$v_i$ is denoted by $Q[i]$ and is multiplied with the costs calculated by a certain cost scheme to obtain the amount that the briber has to pay to $v_i$. We remark that all our tractability results, except the one of Theorem~\ref{thm:weighted-sm-clevel}, hold for arbitrary cost vectors, while all of our hardness results still hold with $ Q[i]=1$ for all $1 \leq i \leq n$.\\

\noindent The $(D,A,C)$-bribery problem is then defined in the following way: 
\begin{quote}
\textsc{$(D,A,C)$-Bribery}\newline
\textbf{Given:}  A profile of $n$ CP-nets over $m$ common binary issues,  a winner determination voting rule $D \in \{SM, OP, OV, OK\}$, a bribery action $A\in \{IV, DV, IV+DV\}$, a budget~$\beta$, a cost scheme $C\in \{\Cequal,\Cflip,\Clevel,\Cany,\Cdist\}$, a cost vector ${\bf Q}\in (\mathbb{N})^n$, and a preferred candidate~$p$. For the voting rule SM, we also require $\mathcal{O}$-legality for a given total order $\mathcal{O}$ over the issues. \newline
\textbf{Question:} Is it possible to bribe the voters to make changes in their cp-statements in such a way that $p$ becomes a (co-)winner of the bribed election, without exceeding~$\beta$?
\end{quote}

In this work, we focus on the non-unique winner model to keep the proofs as simple as possible.
All our results still hold for the unique winner model. Our proofs can easily be adapted for this model, e.g. by adding a tie-breaking voter who initially votes for $p$, cf.\ Faliszewski~\textit{et al.}~\cite{faliszewski2009hard}, or some similar adjustment. Note that in general, tie-breaking is not an easy task and an interesting question of its own (\cite{faliszewski2008copelandties, DBLP:journals/corr/abs-1304-6174, Obraztsova:2011:CVM:2283396.2283450, obraztsova2011ties}).\\

We also consider \textsc{Weighted-$(D,A,C)$-Bribery}, which is defined in the same way, but with weighted voters, which is a typical variant for bribery problems (see the overview of Faliszewski \textit{et al.}~\cite{faliszewski2009hard}).
Moreover, we investigate the computational complexity of weighted and unweighted \textsc{$(D,A,C)$-Negative-Bribery}, which are defined as their non-negative versions, but where the briber is not allowed to bribe voters to vote for his preferred candidate, see also Faliszewski \textit{et al.}~\cite{faliszewski2009hard}, who introduced this variant in the original setting of the bribery problem. \\

\noindent For the polynomial time algorithm we give in Theorem~\ref{thm:weighted-sm-clevel}, we consider the bribery problem as the optimization version of the {\sc Knapsack} problem~\cite{gary1979computers}, which is defined as follows.
\begin{quote}
{\sc Knapsack}\newline
\textbf{Given:} A finite set $U$, for each $u\in U$ a size $s(u)\in \mathbb{Z}^+$ and a value $v(u)\in \mathbb{Z}^+$, and a positive integer $S$. \newline
\textbf{Task:} Find a subset $U' \subseteq U$ such that $\sum_{u \in U'} s(u) \leq S$ and $\sum_{u\in U'} v(u)$ is maximal.
\end{quote}

The {\sc Knapsack} problem is known to be \NP-complete and can be solved in time $\bigO{n\cdot S}$ or $\bigO{n \cdot T}$, where $T:= \sum_{u \in U} v(u)$, with a dynamic programming algorithm, see the work of Dantzig~\cite{dantzig1957discrete}. This running time is pseudo-polynomial since it depends on the representation of the integers $S$ or $T$. We will see that in our case, however, $T$ is bounded by $n\cdot m$. Neither $n$ nor $m$ can be exponential in the input size, since each of the $n$ voters is given by his CP-net and each such CP-net consists of at least $m$ cp-statements.
This will lead to solvability in polynomial time with respect to the input size.\\

\noindent To show \NP-completeness for variants of the bribery problem, we use reductions from the following problems.

\begin{quote}
\SpecialPartition\newline
\textbf{Given:} 
A set $\mathcal{A}$ of (not necessarily different) integers with \newline \hspace*{1.34cm}$\sum_{a\in\mathcal{A}}a=2\psi$.\newline
\textbf{Question:} Is there a subset $\mathcal{A'}\subset \mathcal{A}$ such that $\sum_{a\in \mathcal{A}'} a = \psi$?
\end{quote}

\noindent  
This problem is one of Karp's first 21 \NP-complete problems~\cite{karp1972reducibility}. \\
The following problem is used for the reductions in the negative case. 

\begin{quote}
\textsc{Negative Optimal Lobbying}\newline
\textbf{Given:} An $n \times m$ $0/1$ matrix $\mathcal{E}$, a positive integer $k$, and a $0/1$ vector $x$ of length $m$. (Each row of $\mathcal{E}$ represents an agent. Each column represents a referendum in the election or a certain issue to be voted on by the legislative body. The $0/1$ values in a given row represent the natural inclination of the agent with respect to the referendum questions put to a vote in the election. The vector $x$ represents the outcomes preferred by The Lobby.)\newline
\textbf{Parameter:} $k$ (representing the number of agents to be influenced)\newline
\textbf{Question:} Is there a choice of $k$ rows of the matrix, such that these rows can
be edited, without setting them equal to $x$, so that in each column of the resulting matrix, a majority vote in that column yields the outcome targeted by The Lobby ($=x$)?
\end{quote}

\noindent This problem is the negative version of the \NP-complete \textsc{Optimal Lobbying} problem~\cite{christian2007complexity}. As the following theorem states, it is \NP-complete as well.
In fact, \textsc{Optimal Lobbying} is even \Wzwei-hard  with respect to the number~$k$ of agents to be influenced. This number (which is part of the input) is called a \emph{parameter} in the context of parameterized complexity theory. The class \Wzwei~belongs to the so-called $W$-hierarchy within this theory. We refer to Downey and Fellows~\cite{DF99} for these definitions, or the books of Flum and Grohe~\cite{FG06} or Niedermeier~\cite{niedermeier2006invitation}. To prove membership and hardness with respect to the classes of the $W$-hierarchy, so-called {\it parameterized reductions} are used which have to meet certain requirements on the parameters that are considered. The polynomial-time reduction we give in the proof of the following theorem preserves the parameter~$k$ and therefore is a simple case of a parameterized reduction, implying \Wzwei-hardness with respect to~$k$ for \textsc{Negative Optimal Lobbying}. Informally speaking, this means that the problem is not expected to be solved efficiently even for small values of~$k$.

\begin{theorem}\label{thm:optimal-lobbying-negative}
\textsc{Negative Optimal Lobbying} is \NP-complete and \Wzwei-hard with respect to the number of agents to be influenced.
\end{theorem}

\begin{proof}
We give a reduction from \textsc{Optimal Lobbying}. Starting with an
instance $I_{OL}$ of \textsc{Optimal Lobbying} with an $n \times m$
$0/1$ matrix $\mathcal{E}$, a positive integer $k$, and a $0/1$ vector
$x$ of length $m$, we construct the instance~$I_{NOL}$ of
\textsc{Negative Optimal Lobbying} by extending the vector $x$ of the
preferred outcome of The Lobby by $n$ positions, all of them set to $0$.
The matrix $\mathcal{E}$ is extended by $n$ additional columns as well. We
use the $n\times n$ identity matrix for this extension. Therefore
exactly one entry is $1$ in each of the additional columns, and each
row has one $1$-entry in exactly one new column. All
remaining entries in this extension are set to $0$. The parameter~$k$ is left unchanged.

We make two observations: First, the new matrix still has $n$ rows, so
the majority in the new columns is always $0$. And second, there is
exactly one $1$ in each row in the new columns. So there is no need for
the Lobby to edit anything in the new columns and no row has to be
changed to $x$.
It is not hard to verify that $I_{OL}$ is a yes-instance if and only if
$I_{NOL}$ is a yes-instance.
\end{proof}

\section{Computational Complexity of Bribery Schemes in CP-nets}\label{sec:complexity}

\subsection{Results for the unweighted case}

In this section, we address several cases left open by Mattei \textit{et al.}~\cite{mattei2013bribery}. 
They showed that the unweighted bribery problem is 
solvable in polynomial time for the one-step voting systems $OP$, $OV$, $OK^*$ for several cost schemes. For the cost scheme $\Cany$ and all bribery actions, and for the cost scheme $\Clevel$ with bribery actions $DV$ and $IV+DV$, the computational complexity was still open. In this section, we show that all of these problems are polynomially tractable as well, thus completing the complexity table for $OP$, $OV$, $OK^*$ in the unweighted case. An overview of all results (both from Mattei~\textit{et al.}~\cite{mattei2013bribery} and ours) can be found in Table~\ref{table:results} at the end of this work.\\

\noindent
One task in the proofs of Theorems~\ref{thm:OP_Cany} and~\ref{thm:OV_Cany} which address the open problems consists in finding the~$n$ cheapest candidates in the exponentially large set of candidates. We give a general formulation of this problem in terms of finding the ``smallest subsets'' of a set as follows. By $\mathcal{P}(\mathcal{A})$, we denote the power set of a finite set~$\mathcal{A}$.

\begin{quote}
\summenproblem{}\newline
\textbf{Given:} A finite set $\mathcal{A}=\{a_1, a_2, \dots, a_m\}$, size $s(a_i)\in \mathbb{Z}^+$ for $1 \leq i \leq m$, and a unary coded integer $K \in \mathbb{N},\,  K \leq 2^m =\lvert \mathcal{P}(\mathcal{A})\rvert$.\newline
\textbf{Wanted:} The ``$K$ smallest subsets'' of $\mathcal{A}$, i.e., the first $K$ elements of the power set $\mathcal{P}(\mathcal{A})$, when its elements are sorted ascendingly (ties broken arbitrarily) by the sums of the sizes of their elements.
\end{quote}

In the following, by the {\it size} of a subset, we refer to the sum of the sizes of its elements. The empty set has size zero. The \summenproblem{} hence asks for the $K$ elements of $\mathcal{P}(\mathcal{A})$ with the smallest size.
\begin{example} 
Let $\mathcal{A} = \{a_1, \dots, a_6\}$ with $s(a_1)= s(a_2)= 1, \, s(a_3) = 2, \, s(a_4)= 3, \, s(a_5)= 4, \, s(a_6)= 7$. For a better readability, we use the notation $\mathcal{A} = \{1,\hat{1},2,3,4,7\}$ in this example, where we mark one $1$ with a hat to be able to distinguish between $a_1$ and $a_2$. Then the seven smallest subsets of~$\mathcal{A}$ are the sets $ \{1\}$, $\{\hat{1}\}$, $\{1,\hat{1}\}$, $\{2\}$, $\{1, 2\}$, $\{\hat{1}, 2\}$, $\{3\}$. For $n=8$, there are several possible solutions, each of them consisting of the seven sets $\{1\}$, $\{\hat{1}\}$, $\{1,\hat{1}\}$, $\{2\}$, $\{1, 2\}$, $\{\hat{1}, 2\}$, $\{3\}$ and one of the sets $\{1, \hat{1}, 2\}$, $\{1, 3\}$, $\{\hat{1},3 \}$, $\{4\}$, respectively.
\end{example}

This problem is the enumeration variant (see the recent survey of Eppstein~\cite{eppstein2014kbest} on enumeration variants) of \textsc{K$^\text{\rm th}$ Largest Subset}~\cite{gary1979computers} which was shown to be $\#{}P$-complete by Lawler in 1972~\cite{lawler1972procedure}. The procedure given by Lawler is designed to be applicable to many different problems and therefore is not optimal in terms of the running time. To the best of our knowledge, it was not improved further. Since the applicability of Lawlers procedure to \summenproblem{} can be easily missed, we give a specialized algorithm with the same running time if $m$ equals $K$ and a better running time in every other case\footnote{Lawler's procedure has a running time of $\bigO{Kmc(m)}$, with $c(m)$ being the time required to compute an optimal solution with $m$ variables. In our case $c(m)$ is a constant. Taking into account the time required to find the minimum in the list~\cite[Step 1]{lawler1972procedure} and the creation of the instances~\cite[Step 3]{lawler1972procedure}, this leads to a running time of $\bigO{K(K+m^2)}$ for \summenproblem{}.}.

\begin{theorem}\label{thm:sums}
The problem \summenproblem{} is solvable in $\bigO{\min \{K,m\}\cdot(Km+K)}$ time.
\end{theorem}

\begin{proof}
Since $\mathcal{P(A)}$ contains exponentially many subsets of $\mathcal{A}$, we cannot simply sort them by their sizes if we ask for a polynomial time algorithm. Instead, we proceed in an iterative way as described in Algorithm~\ref{algo}. We assume that the elements of~$\mathcal{A}$ are sorted in ascending order by their sizes, i.e., $s(a_1) \leq s(a_2) \leq \dots \leq s(a_m)$. \\
The idea of the algorithm is the following. Assume we have an array $A$ containing the (up to) $K$ smallest subsets of the set $\mathcal{A}_{i}:=\{a_1, \dots, a_{i}\}$, for some $i \in \{0, 1, \dots, m-1\}$ ($\mathcal{A}_0$ being the empty set), sorted ascendingly by their size. Then we can extend this solution to a solution to \summenproblem{} for the set $\mathcal{A}_{i+1}= \{a_1, \dots, a_{i}, a_{i+1}\}$: We generate an array~$B$ containing the same subsets as $A$ (in the same order). We then adjoin the element $a_{i+1}$ to each of the subsets in $B$. The solution to \summenproblem{} for $\mathcal{A}_{i+1}$ consists of the (up to) $K$ subsets contained in $A$ and $B$ with the smallest size. For finding these we can make use of the sorting of $A$ and $B$, and do a two-way merge~\cite[p. 158]{Knuth1998theart3} until we find the first (up to) $K$ subsets. 

We start this procedure with the solution for $\mathcal{A}_0$, which consists of the empty set only. Then we extend the solution in an iterative way as described above until we obtain the solution for $\mathcal{A}_m=\mathcal{A}$.\\

\noindent\emph{Correctness.} We show that the following statement $\mathcal{I}$ is an invariant for the loop in line~\ref{algo:for} of Algorithm~\ref{algo}:
\begin{quote}
$\mathcal{I}$: Array $A$ contains the (up to) $K$ subsets of $\mathcal{A}_{i}=\{a_1,\dots, a_{i}\}$ with the smallest size in ascending order. 
\end{quote}

\noindent When the algorithm first enters line~\ref{algo:for}, $A$ contains the only subset of $\mathcal{A}_0$, namely $\emptyset$, so $\mathcal{I}$ holds. Now we show that $\mathcal{I}$ still holds after each repetition of the loop.\\
For an iteration $i\geq 1$ of the loop, $A$ stores the (up to) $K$ smallest subsets of the last iteration $i-1$. We create (up to) $K+1$ smallest subsets of $\mathcal{A}_{i}$ which all contain the element $a_{i}$, by subjoining $a_i$ to each of the subsets that we copy from $A$, and by creating the set~$\{a_i\}$ (line~\ref{algo:aclone}). There cannot exist any $a_i$-containing subset~$X$ of $\mathcal{A}_i$ which is smaller than the most expensive one in $B$, because by $\mathcal{I}$, the set $X\setminus \{a_i\}$ has to be contained in $A$. By choosing the (up to) $K$ smallest subsets of $A$ and $B$ to be contained in $A$ for the next iteration, $\mathcal{I}$ holds again at the end of the repetition.\\
After $\min\{K,m\}-1$ iterations of the outer for-loop in line~\ref{algo:for}, $A$ contains the solution to \summenproblem. For $K \geq m$, this is clear because $\mathcal{A}_m = \mathcal{A}$. For $K < m$, the elements $a_{n+1}, \dots, a_m$ cannot be contained in any subset belonging to a solution to \summenproblem, i.e., the solutions to \summenproblem{} for $\mathcal{A}$ and $\mathcal{A}_n$ are identical.\\

\noindent\emph{Running time.} Most of the operations in Algorithm~\ref{algo} can be performed in constant time. 
The for-loop in line~\ref{algo:for2} needs to rewrite at most each entry of $B$. Since $B$ can be represented by a 2-dimensional boolean array with dimension $(K+1)\times m$, this can be done in $\bigO{(K+1)\cdot m}$ time.\\
In line~\ref{algo:if} we compare the sums of the elements of two different subsets. The sum of the elements of each subset can be stored and maintained (line~\ref{algo:aclone}) and can therefore be obtained in constant time for line~\ref{algo:if}. So, the for-loop in line~\ref{algo:for3} can be processed in time $\bigO{K}$.
The copying process in line~\ref{algo:cclone} takes time $\bigO{K\cdot m}$ because of the size of $C$, but one can avoid this by swapping the roles of $A$ and $C$ in each iteration.\\
We hence obtain a running time of $\bigO{\min\{K,m\}\cdot(Km + K)}$.
\end{proof}

\begin{algorithm}[h]
\label{algo}
\SetEndCharOfAlgoLine{}
\caption{Pseudocode for solving \summenproblem{}.}
\KwIn{set $\mathcal{A}=\{a_1,\dots,a_m\}$ with $s(a_i)\in \mathbb{Z}^+$, $K\in \mathbb{N},\,  K \leq 2^m$}
\KwOut{array $A$ containing the $K$ \emph{smallest subsets} of $\mathcal{A}$}
\CommentSty{\% assumption:  $s(a_1) \leq s(a_2) \leq \dots \leq s(a_m)$ }\;
initialize $A$ as an empty array of sets\;
$A[1] := \emptyset$ \CommentSty{\% now $A$ contains only one subset, the empty set}\;
\lnlset{algo:for}{(a)}\For{$i:=1,\dots, \min \{m,K\}$}{
	initialize $B,C$ as empty arrays of sets\;
	\CommentSty{\% create up to $K+1$ new subsets, all containing element $a_i$ }\;
	\CommentSty{\% `$\lvert A\rvert$' denotes the number of subsets (elements) in array $A$}\;
\lnlset{algo:for2}{(b)}	\For{$j := 1, \dots, \lvert A\rvert$}{
		\lnlset{algo:aclone}{(c)}$B[j] := A[j]\cup\{a_i\}$\;
	}
	\CommentSty{\% find up to $K$ elements with smallest size (like a two-way merge)}\; 
	$\textit{indexA}, \textit{indexB} := 1$\;
	\lnlset{algo:for3}{(d)}\For{$j:=1,\dots,\min \{K, \lvert A\rvert+\lvert B\rvert\}$}{
		\CommentSty{\% by `size of $A[k]$' we mean the size of the subset stored in $A[k]$}\;
		\lnlset{algo:if}{(e)}\eIf{$\textit{indexA} \leq \lvert A\rvert$ \textrm{\bf and}  size of $A[\textit{indexA}] < \text{size of }B[\textit{indexB}]$}{
			$C[j] := A[\textit{indexA}]$\;
			$\textit{indexA}\plusplus$\;
		}{
			$C[j] := B[\textit{indexB}]$\;
			$\textit{indexB}\plusplus$\;
		}
	}
	\lnlset{algo:cclone}{(f)}$A := C$\;
}
\Return $A$\;
\vspace*{0.1cm}
\end{algorithm}

We are now ready to prove some of the open questions that were stated by Mattei~\textit{et al.}~\cite{mattei2013bribery}.

\begin{theorem}\label{thm:OP_Cany}
\textsc{$(OP,A,\Cany)$-Bribery} is in \Pclass{} if $A\in \{IV, DV, IV+DV\}$.
\end{theorem}

\begin{proof} 

In this proof, for ease of presentation, the briber's preferred candidate is called~$c_1$. Moreover, for every voter $v_i$, we denote by $\costs_i(c_j)$ the costs to bribe $v_i$ such that $c_j$ becomes his favorite candidate.
We show Theorem~\ref{thm:OP_Cany} by construction of a flow network that can be solved with a minimum cost flow algorithm (which also maximizes the flow in polynomial time, cf. Ahuja~\textit{et al.}~\cite{ahuja1993network}), similarly as initially proposed by Faliszewski~\cite{faliszewski2008nonuniform}. This method has proven to be useful for showing tractability for many voting problems and is also used by Mattei~\textit{et al.}~\cite{mattei2013bribery}.
For each~$r\in \{1,\dots, n \}$, we check if the voters can be bribed with budget~$\beta$ such that the preferred candidate~$c_1$ wins with a score of~$r$ votes. If this is possible for at least one $r$, then we accept, otherwise we reject.  For each $1\leq r \leq n$, we define a flow network consisting of a source node $s$, a target node $t$, and the following sets of nodes and directed edges. All edges have costs $0$ and capacity $1$, unless specified otherwise. 
\begin{description}
\item[V(oters):] For every voter $v_i$ we create a node $v_i$ and an edge $(s,v_i)$. The set of nodes constructed in this step is called~$V$.
\item[B(ribery):] For each voter $v_i$, we determine the set consisting of the $n$ cheapest candidates (the candidates for which the briber has to pay least if he bribes $v_i$ to make them his new top candidate) and add $c_1$ to this set (if not present within those $n$ candidates). For each candidate $c_j$ of this set, we create a node $c^i_j$. Note that at least one candidate with zero costs is contained in this set (for instance the top candidate). For each such node $c_j^i$ we create an edge $(v_i,c_j^i)$ with costs $\costs_i(c_j)$. The flow on these edges tells the briber who and how he has to bribe. We call the set of all nodes created in this step $B$.
\item[C'(ollection):] We create a node $c_j^*$ if there exists a $j\in [1,m]$ with  $c^i_j \in B, i \in [1,n]$. For each node $c^i_j \in B$, we add the edge $(c_j^i, c_j^*)$. We call the set of all nodes created in this step $C'$.
\item[Gadget node:] We create one gadget node $g$ and add the edges $(c_j^*, g)$ for all nodes $c_j^* \in C'\setminus \{c_1^*\}$. These edges all have capacity $r$. Finally we add the edge $(c_1^*,t)$ with capacity $r$ and the edge $(g,t)$ with capacity $n-r$.
\end{description}

\noindent An example for a small instance consisting of two voters is given in Figure~\ref{fig:op_Cany}.

\begin{figure}[h!]
\centering
\begin{tikzpicture} [->,thick,align=center,xscale=2,yscale=0.6]
  \path[every node/.style={minimum size=0.5cm}]
  	   (0, 0)   node (s)  {$s$}
       (1, 2)   node (l1v1) {$v_1$}
       (1,-2)   node (l1v2) {$v_2$}
       (2,0.75) node (l2o1) {$c_4^1$}
       (2,1.75) node (l2o2) {$c_3^1$}       
       (2,2.75) node (l2o3) {$c_2^1$}
       (2,3.75) node (l2o4) {$c_1^1$}       
       (2,-0.75) node (l2u1) {$c_1^2$}
       (2,-1.75) node (l2u2) {$c_5^2$}      
       (2,-2.75) node (l2u3) {$c_3^2$}
       (2,-3.75) node (l2u4) {$c_4^2$}
       (3,2) node (l3c1) {$c_1^*$}
       (3,1) node (l3c2) {$c_2^*$}
       (3,0) node (l3c3) {$c_3^*$}
       (3,-1) node (l3c4) {$c_4^*$}
       (3,-2) node (l3c5) {$c_5^*$}
       (4,-0.5)   node   (g)  {$g$}
       (5,0) node (t) {$t$}
       (1,4.75) node (a) {$V$}
       (2,4.75) node (b) {$B$}
       (3,4.75) node (b2) {$C'$};
  \path[every node/.style={font=\sffamily\small}]
    (s)		edge node[above] {} (l1v1)
    			edge node[above] {} (l1v2)
    	(l1v1)	edge node[above, pos=0.7] {$3$} (l2o1)
    			edge node[above, pos=0.7] {$2$} (l2o2)
    			edge node[above, pos=0.7] {$0$} (l2o3)
    			edge node[above, pos=0.7] {$18$} (l2o4)
    	(l1v2)	edge node[above, pos=0.7] {$6$} (l2u1)
    			edge node[above, pos=0.7] {$0$} (l2u2)
    			edge node[above, pos=0.7] {$1$} (l2u3)
    			edge node[above, pos=0.7] {$5$} (l2u4)
    	(l2o1)	edge node[above] {} (l3c4)
    	(l2o2)	edge node[above] {} (l3c3)
    	(l2o3)	edge node[above] {} (l3c2)
    	(l2o4)	edge node[above] {} (l3c1)
    	(l2u1)	edge node[above] {} (l3c1)
    	(l2u2)	edge node[above] {} (l3c5)
    	(l2u3)	edge node[above] {} (l3c3)
    	(l2u4)	edge node[above] {} (l3c4)
    	(l3c1)	edge[dashed] node[above] {$r$} (t)
    	(l3c2)	edge[dashed] node[above, pos=0.3] {$r$} (g)
    	(l3c3)	edge[dashed] node[above, pos=0.3] {$r$} (g)
    	(l3c4)	edge[dashed] node[above, pos=0.3] {$r$} (g)
    	(l3c5)	edge[dashed] node[above, pos=0.3] {$r$} (g)
    	(l3c2)	edge[dashed] node[above, pos=0.3] {$r$} (g)
    	(g)	edge[dashed] node[above, pos=0.3] {$n-r$} (t);
   
\end{tikzpicture}

\caption{Example for a flow network of an election with two voters $v_1$ and $v_2$. All edges have costs $0$ and capacity $1$. The only exceptions are the dashed edges that have a different capacity stated, and the connecting edges $(v_i,c^i_j)$ between $V$ and $B$  that have the costs to bribe $v_i$ to vote for $c_j$.}\label{fig:op_Cany}
\end{figure}
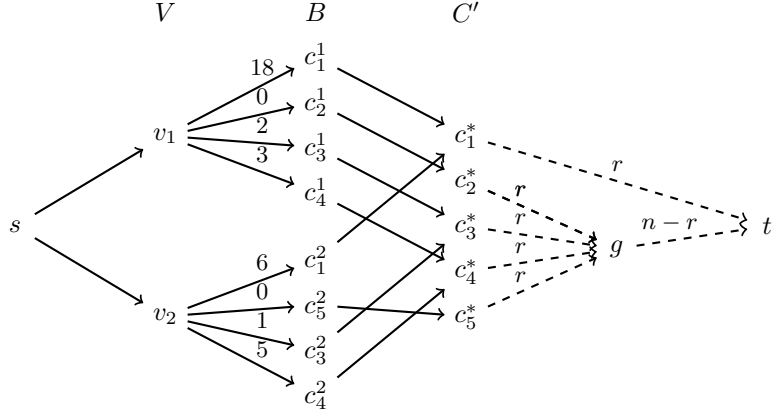

As explained above, we create $n$ such networks, each differing only in the value of~$r$ which is the maximum score all candidates can reach and the preferred candidate $c_1$ will win with, and solve the corresponding flow problem with a minimum cost flow algorithm. The resulting flow will always be $n$, but the costs may differ. 
We then choose the solution with the lowest costs, and if they do not exceed the given budget, we found an optimal one. Otherwise there is none. 
The constructed networks have a number of nodes that is polynomial in $n$ and we are only building up to~$n$ of them. This results in a polynomial running time of the overall algorithm.

For the correctness, observe that there is a flow of value~$n$ on the constructed network if and only if the bribery problem can be solved: 
Assume there exists a flow of value~$n$ in the network. Since all capacities and costs are integral, the flow is integral as well~\cite{ahuja1993network}. 
The edges connecting nodes in $C'$ with $g$ and $t$, respectively, guarantee that a flow of value $r$ is going through vertex $c_1^*$, hence ensuring that the preferred candidate~$c_1$ obtains $r$ points, and that no flow of value greater than~$r$ leaves any vertex $c_j^*\in C'\setminus\{c_1^*\}$, hence ensuring that no candidate can end up with more than $r$ points, making $c_1$ a (co-)winner. This is achieved if the briber changes those votes corresponding to the edges connecting $V$ and $B$: For each node~$v_i$, there is one edge $(v_i, c_j^i)$ connecting $v_i$ to a node $c_j^i$ of the set $B$, carrying a flow of value 1. This edge tells the briber that $v_i$ has to be bribed in such a way that $c_j$ becomes his new top candidate. Therefore, there exists a successful bribery in which $c_1$ wins with $r$ votes.  

Conversely, given a bribery that makes $c_1$ a winner of the election without exceeding budget $\beta$, we can construct a flow in the following way: Since $c_1$ is a winner of the bribed election, he must have obtained a total of $r$ points for some $1\leq r \leq n$. So we choose the network with capacity $r$ on the edges $(c_j^*,g)$ and $(c_1^*,t)$. We set the flow on the edges $(s,v_i)$ with $v_i \in V$ to~1. 
We distinguish two types of voters: 
If after the bribery, candidate~$c_j$ is on the top position for voter~$v_i$ and the node $c_j^i$ is created for $B$ (meaning that $c_j$ is among the cheapest $n$ candidates for $v_i$), we set the flow on the edges $(v_i, c_j^i)$ and $(c_j^i, c_j^*)$ to 1, otherwise we call the voter $v_i$ an {\it overpaid} voter and ignore him for now. Note that $c_1^i$ is a node in $B$ for each~$i$. Having set as much flow on the edges between nodes of $V$ and $C'$ as this rule allows, we carry on preserving the flow from nodes in $C'$ to $t$. This is easy, because there is exactly one path from each node in $C'$ to $t$. 
For each voter $v_k$ in the set of overpaid voters, the flow on the edges between the node $v_k$ and nodes of $B$ as well as on the edges leaving these nodes of $B$ has not been set yet. For each such voter $v_k$, we just have to find an \emph{augmenting path} (cf.\ the Edmond-Karp algorithm~\cite{edmonds1972theoretical} for the Ford-Fulkerson algorithm~\cite{ford1956maximal} for solving the flow problem) from $v_k$ to $t$. There are at least $n$ nodes in $C'$ and $n$ different nodes $c^k_j$ for each $v_k$, so by the pigeon hole principle, there exists at least one such path, and the cost of each such path is lower than the amount paid by the briber to bribe $v_k$. Therefore there is a flow of value~$n$ for this network with costs of at most~$\beta$. 

There are two more things to be observed.
First, as described by Mattei~\textit{et al.}~\cite[Theorem 8]{mattei2013bribery}, for each voter, it suffices to choose the $n$ cheapest candidates when constructing the set $B$ in the network instead of including all $2^m$ possible candidates: There are $n-1$ other voters who can give one vote each to a candidate. By the pigeonhole principle, it is not possible that each of the $n$ cheapest candidates for a voter already has a score of $r$ votes, hence in an optimal solution, no additional (more expensive) candidate has to be bribed.
Second, the selection of the nodes in $B$ corresponds to solving the \summenproblem{} problem, where the sizes of the elements in $\mathcal{A}$ correspond to the allowed individual costs for flipping every issue of voter~$v_i$. This can be done with the algorithm presented in the proof for Theorem~\ref{thm:sums}. Each subset $\mathcal{F}$ belonging to the solution to the \summenproblem{} algorithm represents one candidate, where the size of the elements in $\mathcal{F}$ are the costs of the required flips to bribe~$v_i$ to vote for this candidate. 
To cope with the different bribery actions, one can simply adjust $\mathcal{A}$. There are no further changes needed.
\end{proof}

Theorem~\ref{thm:OP_Cany} implies that $(OP,A,C_\textsc{level})$-bribery can be solved in polynomial time as well, because $C_\textsc{level}$  is a special case of $C_\textsc{any}$. 
In fact, the same holds for the cost schemes $\Cflip$ and $\Cdist$, but the complexity for those was already shown by Mattei~\textit{et al.}~\cite{mattei2013bribery}.

\begin{corollary}\label{cor:op}
$(OP,A,\Clevel)$-Bribery is in \Pclass{} if $A\in \{IV, DV, IV+DV\}$.
\end{corollary}

Theorem~\ref{thm:OP_Cany} can be extended to the voting system $OK$ for $\mathcal{O}$-legal profiles, when $k$ is a power of~$2$, which is called $OK^*$ then. To show this, we use Lemma~$1$ from Mattei~\textit{et al.}:

\begin{lemma}[Mattei \textit{et al.}~\cite{mattei2013bribery}]\label{lemma:mattei1}
Given an acyclic CP-net $X$ and a constant linearization scheme across all agents and $k=2^j$, for some $j\in \mathbb{N}$, the top $k$ outcomes of $X$ are all the outcomes differing from the top one on the value of exactly $j$ issues. Moreover, given two CP-nets in the same profile and with the same top element, they have the same top $k$ elements.
\end{lemma}

We recall that existence of a constant linearization scheme across all agents means $\mathcal{O}$-legality of the profile for some given strict total order $\mathcal{O}$.
Lemma~$1$ implies that any \textsc{$(OK^*,A,C)$-Bribery} instance can be treated as a \textsc{$(OP,A,C)$-Bribery} instance, for any cost scheme $C$ and bribery action $A$, by ignoring the last $\log_2 k$ issues.

\begin{corollary}\label{cor:OK}
\textsc{$(OK^*,A,C)$-Bribery} with bribery action $A\in\{IV,DV,IV+DV\}$ and cost scheme $C\in\{\Cany,\Clevel\}$ is in \Pclass{}. 
\end{corollary}

Similarly as in the proof of Theorem~\ref{thm:OP_Cany}, one can show tractability for the voting system $OV$ with cost scheme $\Cany$.

\begin{theorem}\label{thm:OV_Cany}
\textsc{$(OV,A,C_\textsc{any})$-Bribery} is in \Pclass{} if $A\in \{IV, DV, IV+DV\}$.
\end{theorem}

\begin{proof}
We distinguish two cases: $n < 2^m$ and $n \geq 2^m$. The first case is straightforward, the second case can be translated into a network flow problem. The solution flow indicates who has to be bribed and how such that the preferred candidate $p$ wins the election.
\begin{description}
\item[$n<2^m$:] By the pigeonhole principle, there must exist at least one candidate who got no veto at all and is therefore the winner of the election. The preferred candidate $p$ has to get rid of all his vetoes to be one of the winners. Therefore the briber has to bribe every voter who vetoes against $p$, by choosing the cheapest bribery possible.
\item[$n \geq 2^m$:] We construct a flow network similar to the one in the proof of Theorem~\ref{thm:OP_Cany}. Again, the preferred candidate~$p$ of the briber is called $c_1$.
We take a source $s$ and a target $t$ for the network and construct the following sets of vertices and edges. All edges have costs $0$ and capacity $1$, unless specified otherwise.  
\begin{description}
\item[General:] 
\item[V(oters):] For every voter $v_i$ we create a node $v_i$ and an edge $(s,v_i)$. These edges each have capacity $2^m-1$.
\item[A(pproval):] For each voter $v_i$ and each candidate $c_j$ who gets approved by $v_i$, we create a node $\hat{c}_j^i$ together with the edge $(v_i, \hat{c}_j^i)$.
\item[B(ribery):] For each voter $v_i$ and every candidate $c_i$, we create a node $c_j^i$. Let $c_x$ be the candidate $v_i$ vetoes against. For each candidate $c_j \neq c_x$ we then create the edge $(\hat{c}^i_j, c^i_j)$ with costs~$0$. If it is allowed (depending on the bribery action) to bribe voter~$v_i$ to veto for candidate~$c_j$, we additionally create the edge $(\hat{c}^i_j, c^i_x)$ with the costs to bribe~$v_i$ this way. Flow in the solution on those additional edges indicates the veto is bribed from $c_x$ to $c_j$. Note that the costs for the edge $(\hat{c}^i_j, c^i_x)$ can be calculated in time polynomial in $m$.
\item[C'(ollection):] For each candidate $c_j$ we create a node $c_j^*$. For each node $c^i_j$ created in the previous step, we add the edge $(c_j^i, c_j^*)$.
\item[Gadget-node:] We create one gadget-node $g$ and add the edges $(c_j^*, g)$ for all nodes $c_j^*$ created in the previous step except for $c_1^*$. These edges all have capacity $r$. Finally we add the edge $(c_1^*,t)$ with capacity $r$ and the edge $(g,t)$ with capacity $n\cdot(2^m-1)-r$.
\end{description}
\end{description}

An example is given in Figure~\ref{fig:unweighted_OV_Cany}.

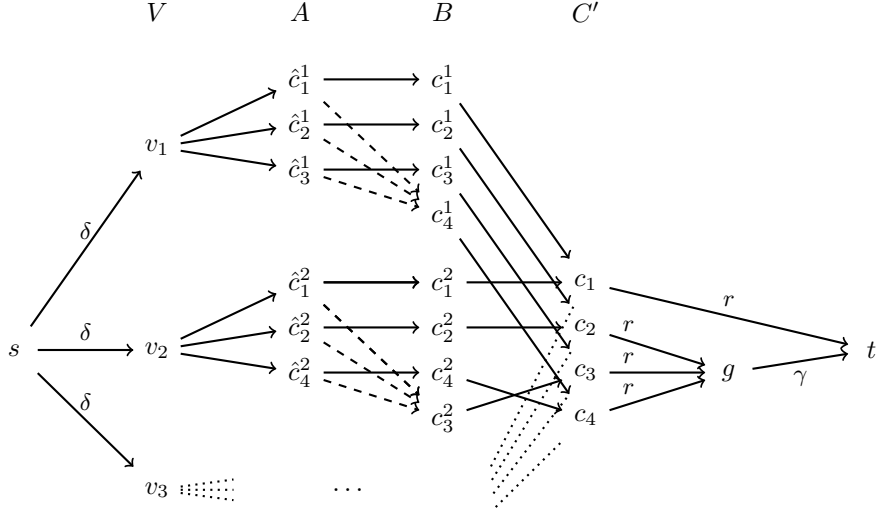
\begin{figure}[h!]
\centering
\begin{tikzpicture} [->,thick,align=center,xscale=1.9,yscale=0.6] 
  \path[every node/.style={minimum size=0.6cm}]
  	   (0, 0)	node (s)  {$s$}
       (1, 4.5)	node (V1) {$v_1$}
       (1,0)		node (V2) {$v_2$}
       (1,-3.1)	node (V3) {$v_3$}
       (2,6)		node (A11) {$\hat{c}_1^1$}
       (2,5)		node (A12) {$\hat{c}_2^1$}
      (2,4)		node (A13) {$\hat{c}_3^1$}
       (2,1.5)	node (A21) {$\hat{c}_1^2$}
       (2,0.5)	node (A22) {$\hat{c}_2^2$}
      (2,-0.5)	node (A23) {$\hat{c}_4^2$}
       (3,6)		node (B11) {$c_1^1$}
       (3,5)		node (B12) {$c_2^1$}
      (3,4)		node (B13) {$c_3^1$}
       (3,3)		node (B14) {$c_4^1$}
       (3,1.5)	node (B21) {$c_1^2$}
       (3,0.5)	node (B22) {$c_2^2$}
       (3,-0.5)	node (B23) {$c_4^2$}
       (3,-1.5)	node (B24) {$c_3^2$}
       (4,1.5)	node (Cc1) {$c_1$}
       (4,0.5)	node (Cc2) {$c_2$}
       (4,-0.5)	node (Cc3) {$c_3$}
       (4,-1.5)	node (Cc4) {$c_4$}
       (5,-0.5)	node (g) {$g$}
       (6, 0)	node (t) {$t$}
       (1,7.5)	node (v) {$V$}
       (2,7.5)	node (a) {$A$}
       (3,7.5)	node (b) {$B$}
       (4,7.5)	node (c) {$C'$}
       (1.7,-2.8)	node (v31) {}
       (1.7,-3.1)	node (v32) {}
       (1.7,-3.4) node (v33) {}
       (2.35,-3.1)	node (dots) {\dots}
       (3.3,-2.8)	node[minimum size=0cm] (v31_2) {}
       (3.3,-3.1)	node[minimum size=0cm] (v32_2) {}
       (3.3,-3.4)	node[minimum size=0cm] (v33_2) {}
       (3.3,-3.7)	node[minimum size=0cm] (v34_2) {};
   
  \path[every node/.style={font=\sffamily\small}]
    		(s)		edge node[above] {$\delta$} (V1)
	    			edge node[above] {$\delta$} (V2)
    				edge node[above] {$\delta$} (V3)
	    	(V1)		edge node[above] {} (A11)
    				edge node[above] {} (A12)
				edge node[above] {} (A13)
	    	(V2)		edge node[above] {} (A21)
    				edge node[above] {} (A22)
    				edge node[above] {} (A23)
    		(A11)	edge node[above] {} (B11)
    				edge[dashed] node[above] {} (B14)
    		(A12)	edge node[above] {} (B12)
    				edge[dashed] node[above] {} (B14)
    		(A13)	edge node[above] {} (B13)
    				edge[dashed] node[above] {} (B14)
    		(A21)	edge node[above] {} (B21)
    				edge[dashed] node[above] {} (B24)
		(A21)	edge node[above] {} (B21)
    				edge[dashed] node[above] {} (B24)
    		(A22)	edge node[above] {} (B22)
    				edge[dashed] node[above] {} (B24)
    		(A23)	edge node[above] {} (B23)
    				edge[dashed] node[above] {} (B24)
    		(B11)	edge node[above] {} (Cc1)
    		(B12)	edge node[above] {} (Cc2)
    		(B13)	edge node[above] {} (Cc3)
	    	(B14)	edge node[above] {} (Cc4)
    		(B21)	edge node[above] {} (Cc1)
    		(B22)	edge node[above] {} (Cc2)
    		(B23)	edge node[above] {} (Cc4)
	    	(B24)	edge node[above] {} (Cc3)
	    	(v31_2)	edge[-,dotted] node[above] {} (Cc1)
    		(v32_2)	edge[-,dotted] node[above] {} (Cc2)
    		(v33_2)	edge[-,dotted] node[above] {} (Cc3)
	    	(v34_2)	edge[-,dotted] node[above] {} (Cc4)
	    	(V3) 	edge[-,dotted] node[above] {} (v31)
    		(V3)		edge[-,dotted] node[above] {} (v32)
	    	(V3)		edge[-,dotted] node[above] {} (v33)
	    	(Cc1)	edge node[above] {$r$} (t)
    		(Cc2)	edge node[above, pos=0.2] {$r$} (g)
    		(Cc3)	edge node[above, pos=0.2] {$r$} (g)
    		(Cc4)	edge node[above, pos=0.2] {$r$} (g)
    		(g)		edge node[below] {$\gamma$} (t);
   
\end{tikzpicture}
\caption{Example of a flow network created from an election with three voters and $m=2$ resulting in $4$ candidates. Only the first two voters are shown in detail. Voter $v_1$ casts his veto against $c_4$, while voter $v_2$ casts his veto against $c_3$. The costs of all edges are $0$ except for the dashed $(\hat{c}^i_j,c^i_x)$ ones which carry the costs to bribe $v_i$ in a way to veto against $c_j$ (instead of $c_x$). The capacities are all $1$ except for the edges where a different value is specified in the graph. Here, $\delta=2^m-1$ and $\gamma=n\cdot(2^m-1)-r$, where~$r$ ranges from 1 to $n$ for the $n$ networks constructed.}\label{fig:unweighted_OV_Cany}
\end{figure}

We proceed as in the proof of Theorem~\ref{thm:OP_Cany}. We again construct a flow network for each~$r$ with  $1\leq r\leq n\cdot (2^m-1)$ and take the solution with a flow of $n\cdot(2^m-1)$ and minimal costs as the solution to the bribery problem. Again, solving the flow problem is equivalent to solving the bribery problem. We remark that the case $n=2^m$ can also be handled with a smaller network.
\end{proof}

As before, $C_\textsc{level}$ is a special case of $C_\textsc{any}$, implying  the following corollary.

\begin{corollary}\label{cor:ov_Clevel}
$(OV,A,C_\textsc{level})$-Bribery is in \Pclass{} if $A\in \{IV, DV, IV+DV\}$.
\end{corollary}

\subsection{Results for the weighted case}\label{subsec:weighted-pos}

In this section, we analyze the weighted version of the bribery problem. Mattei \textit{et al.}~\cite{mattei2013bribery} also consider the case of weighted voters for the voting system $SM$ for which they obtain \NP-completeness. To the best of our knowledge, no result in the weighted case for $OP$, $OV$, and $OK$ has been published yet. We show for almost all of these remaining systems that the bribery problem is \NP-complete as well by reductions from the {\sc Partition} problem.

The idea in all of the reductions in this section is the following.
We construct an election in which one candidate called $c_2$ is the winner with a total of $2\summe$ points which he obtains by a set of voters whose weights correspond to the elements in the set to be partitioned. The preferred candidate~$p$ obtains~$\summe$ points. The instance is built in such a way that for the briber, the only way to make $p$ win consists in  transferring points from $c_2$ to another candidate $c_3$, thus solving the {\sc Partition} problem, such that $c_2, \, c_3$ and $p$ end up with the same score of~$\summe$ as co-winners.

\begin{theorem}\label{thm:weighted-op}
\textsc{Weighted-$(OP,A,C)$-Bribery} is \NP-complete with bribery action $A\in\{IV,DV\}$ and cost scheme $C\in \{\Cequal,\Cany, \Clevel, \Cdist, \Cflip\}$.
\end{theorem}

\begin{proof}
We show Theorem~\ref{thm:weighted-op} by reduction from \SpecialPartition{}. Let $I_P$ be an instance of \SpecialPartition{} with $\mathcal{A}= \{a_1, \dots, a_{\ell}\}$ and $\sum_{a\in \mathcal{A}} a = 2\summe$. The instance $I_B$ of \textsc{Weighted-$(OP,A,C)$-Bribery} that is constructed from $I_P$ has two issues~$X$ and~$Y$ with domains~$\{x,\overline{x}\}$ and ~$\{y,\overline{y}\}$, respectively, and therefore four candidates. For the bribery action $DV$ these are the following:

\begin{description}
\item[$xy$] The preferred candidate $p$ of the briber. He has one~$\summe$-weighted voter voting for him.
\item[$\overline{x}y$] The candidate $c_2$, who starts as the winner with $2\summe$ votes.
\item[$\overline{x}\overline{y}$] The candidate $c_3$, for whom no one votes initially, but who will win together with $p$ and $c_2$ if there exists a partition.
\item[$x\overline{y}$] The candidate $c_4$ who is not important for the case of bribery action $DV$ but for the case of bribery action $IV$.
\end{description}

\noindent We achieve this by constructing the following voters:

\begin{center}
\begin{tabular}{ccccc}
\toprule 
voter~ & top candidate~ & dependency~ & weight~ & voting for\\
\cmidrule{1-5}
$v_1$ & $xy$ & none & $\summe$ & $p$  \\ 
$v_2$ & $\overline{x}y$ & $X\rightarrow Y$ & $a_1$ & $c_2$ \\ 
$v_3$ & $\overline{x}y$ & $X\rightarrow Y$ & $a_2$ & $c_2$ \\ 
\multicolumn{5}{c}{\vdots}\\ 
$v_{\ell+1}$ & $\overline{x}y$ & $X\rightarrow Y$ & $a_{\ell}$ & $c_2$ \\
\bottomrule
\end{tabular}
\end{center}

\noindent Note that the briber is only allowed to bribe the dependent issue $Y$ for every voter except $v_1$, so he cannot transfer points from $c_2$ to $p$. He can just split the points of $c_2$ among $c_2$ and $c_3$, in which case they win together with $p$. In every other case~$p$ loses. \\
More formally, we show that $I_P$ is a yes-instance if and only if $I_B$ is a yes-instance.
Let $\mathcal{A'}\subset \mathcal{A}$ be a solution to $I_P$. Then the briber bribes those voters whose weights are the elements of~$\mathcal{A'}$. This results in a score of~$\summe$ for $p, \, c_2$ and $c_3$, so $I_B$ is a yes-instance.\\
Conversely, let $I_B$ be a yes-instance, implying there is a bribery that makes $p$ a (co-)winner of the given election. Since the briber can only bribe the dependent issue $Y$, he can make the voters $v_2, \dots , v_{\ell+1}$ vote for the candidate $\overline{x}\overline{y} = c_3$ only. The only possible way to make $p$ a (co-)winner is hence to split the points of $c_2$ among $c_2$ and $c_3$ so that they both end up with a score of~$\summe$. The weights of the voters that have to be bribed to make them vote for $c_3$ are then the elements of the subset~$\mathcal{A}' \subset \mathcal{A}$ that solves $I_P$.\\
Since no costs are involved in this reduction, it works for every cost scheme considered, assuming an unlimited budget.

For the case $IV$, we can use the same construction. In this case, the briber can only bribe the independent issue~$X$, hence split the points of $c_2$ among $c_2$ and $c_4= x\overline{y}$, which is again the only possible way to make $p$ a winner. 
\end{proof}

Since $OP$ is a special case of $OK$, we obtain the following corollary.

\begin{corollary}\label{cor:weighted_OK}
\textsc{Weighted-$(OK,A,C)$-Bribery} is \NP-complete for bribery action $A\in\{IV,DV\}$ and cost scheme $C\in \{\Cequal,\Cany, \Clevel, \Cdist, \Cflip\}$.
\end{corollary}

For the voting rule $OV$ we can also reduce from the partition problem when we restrict the briber to the $DV$ actions only and obtain the following theorem. The cases of $IV$ and $IV+DV$ are not covered yet.

\begin{theorem}\label{thm:ov-weighted}
\textsc{Weighted-$(OV,DV,C)$-Bribery} is \NP-complete for each cost scheme $C\in \{\Cequal,\Cany, \Clevel, \Cdist, \Cflip\}$.
\end{theorem}

\begin{proof}
We show Theorem~\ref{thm:ov-weighted} by reduction from \SpecialPartition{} once more. Let $I_P$ be an instance of \SpecialPartition{} with $\mathcal{A}= \{a_1, \dots, a_{\ell}\}$ and $\sum_{a\in \mathcal{A}} a = 2\summe$. The instance $I_B$ of \textsc{Weighted-$(OV,DV,C)$-Bribery} that is constructed from $I_P$ has two issues~$X$ and~$Y$ with domains~$\{x,\overline{x}\}$ and ~$\{y,\overline{y}\}$, respectively, and therefore four candidates. For the bribery action $DV$ these are the following:

\begin{description}
\item[$xy$] The preferred candidate $p$ of the briber. There is one voter with weight~$\summe$ casting his veto against him.
\item[$\overline{x}y$] The candidate $c_2$, who initially receives $2\summe$ vetos.
\item[$\overline{x}\overline{y}$] The candidate $c_3$, against whom no one casts a veto initially, but who will win together with $p$ and $c_2$ if there exists a partition.
\item[$x\overline{y}$] An unimportant clone~$u$ of $p$.
\end{description}

\noindent We achieve this by constructing the following voters:

\begin{center}
\begin{tabular}{ccccc}
\toprule 
voter~ & top candidate~ & dependency~ & weight~ & vetoes agains\\
\cmidrule{1-5}
$v_1$ & $\overline{x}\overline{y}$ & none & $\summe$ & $p$  \\
$v_2$ & $\overline{x}y$ & none & $\summe$ & $u$\\ 
$v_3$ & $xy$ & $X\rightarrow Y$ & $a_1$ & $c_2$ \\ 
$v_4$ & $xy$ & $X\rightarrow Y$ & $a_2$ & $c_2$ \\ 
\multicolumn{5}{c}{\vdots}\\ 
$v_{\ell+2}$ & $xy$ & $X\rightarrow Y$ & $a_{\ell}$ & $c_2$ \\
\bottomrule
\end{tabular}
\end{center}

\noindent Note that the briber is only allowed to bribe the dependent issue~$Y$ for every voter except $v_1,\, v_2$, so he can only transfer vetos from $c_2$ to $c_3$. The best he can do is to split the vetos of $c_2$ among $c_2$ and $c_3$, in which case they win together with $p$ and $u$. In every other case $p$ loses. As in the proof of Theorem~\ref{thm:weighted-op}, the two instances are equivalent. As no costs are involved here, the reduction works for every cost scheme considered, assuming an unlimited budget.
\end{proof}
We now investigate the weighted case for the voting rule $SM$. We found that the situation here has to be analyzed in more detail: Hardness of the weighted version of the bribery problem can be caused just by the use of the cost vector. \\
Mattei \textit{et al.}~\cite[Theorem 6]{mattei2013bribery} showed that \textsc{Weighted-$(SM,A,C)$-Bribery} is \NP-complete for all bribery actions and for all cost schemes they considered. 
However, for the cost schemes $C\in\{\Clevel,\Cdist,\Cflip,\Cany\}$, their reduction implicitly makes use of the cost vector ${\bf Q}$. Hence, the computational hardness originates from the individual costs for each voter. \\
This is not the case for $\Cequal$, where \NP-completeness already follows from the result in the unweighted case, without any requirement on individual costs.\\
Mattei \textit{et al.}~observe that for the cost scheme $\Cflip$, if the cost vector ${\bf Q}$ contains only ones, \textsc{Weighted-$(SM,A,\Cflip)$-Bribery} is in \Pclass{} with bribery action $A\in\{IV,DV,IV+DV\}$ (\cite[Theorem~7]{mattei2013bribery}).
We found the same property for \textsc{Weighted-$(SM,A,\Clevel)$-Bribery}, which is also solvable in polynomial time for all bribery actions considered in this work if the cost vector contains only ones.

\begin{theorem}\label{thm:weighted-sm-clevel}
\textsc{Weighted-$(SM,A,\Clevel)$-Bribery} with bribery action $A\in\{IV,DV,IV+DV\}$ is, for the special case that $Q[i]=1$ for each voter $v_i$, solvable in time $\bigO{n^2 m^2}$,  .
\end{theorem}

\begin{proof}
For the voting system $SM$, it is sufficient to look at the issues one after another, since they do not affect each other, even with dependencies. For each issue, a certain amount of weighted voters has to be bribed to flip the value of the corresponding issue. The briber just needs to know which of the voters are the cheapest ones to bribe. We first start with bribery action $IV+DV$.

Let $V$ be the set of all voters, let~$w_i$ be the weight of voter~$v_i$, $W=\sum_{i=1}^n w_i$ the sum of all weights, and $\text{costs}(v_i)$ the cost to bribe voter~$v_i$ in the considered issue. The set of all voters who vote in the given issue for the same value that is taken by~$p$ in this issue is called~$\mathcal{G}$ for the \emph{good} voters, and $\mathcal{B} = V\setminus \mathcal{G}$ denotes the set of the remaining \emph{beneficiary} voters. The briber is interested in finding a subset $\mathcal{S'}\subseteq \mathcal{B}$, such that $\sum_{v_i \in \mathcal{S'}\cup \mathcal{G}} w_i > \lfloor\frac{W}{2}\rfloor$ and $\sum_{v_i\in \mathcal{S'}}\text{costs}(v_i)$ is minimal. This is the same as finding a subset $\mathcal{S}\subseteq \mathcal{B}$, such that $\sum_{v_i\in\mathcal{S}} w_i \leq \lfloor \frac{W}{2}\rfloor$ and $\sum_{v_i\in \mathcal{S}}\text{costs}(v_i)$ is maximal. 
The latter problem is the same as solving the \textsc{Knapsack} problem on the set $\mathcal{B}$ where the bribery costs $\text{costs}(v_i)$ correspond to the values, and the weights of the voters correspond to the size of the items to be packed in the knapsack, respectively. The crucial observation here is that with $\Clevel$, the cost for flipping the value of the relevant cp-statement is bounded by the number~$m$ of issues, as there cannot be more levels than issues by definition, hence $\text{costs}(v_i)\leq m$ for all $v_i \in V$. As mentioned in Section~\ref{sec:prel}, one can solve the \textsc{Knapsack} problem with dynamic programming in time $\bigO{n \cdot T}$, where $T$ is the sum of the values of all objects. In our case, $T$ is the sum $\sum_{i=1}^n \text{costs}(v_i)$, hence bounded by $n\cdot m$, so that the running time of the dynamic programming algorithm proposed by Dantzig (see~\cite{dantzig1957discrete}) is polynomial. In detail, we implement Dantzig's solution as follows:

We define an $n \times \sum_{v_i \in \mathcal{B}} \text{costs}(v_i)$ matrix $D$ with entries $d_{i,j}$ being the minimum weight of all subsets $\mathcal{S}\subseteq \{v_1,\dots,v_i\}$ with $\sum_{v_i\in \mathcal{S}} \text{costs}(v_i)=j$. Setting $d_{i,0}=0$ for all $i$, $d_{-1,j}=\infty$ for $j> 0$, and $d_{i,j}=\infty$ for  $j<0$, we can compute all entries of $D$ recursively by  
\begin{align*}
d_{i,j} = \min \begin{cases}
 d_{i-1,j},\\
 d_{i-1,j-\text{costs}(v_i)}+w_j,
\end{cases}
\end{align*}
with $w_i$ and $\text{costs}(v_i)$ being the weight of voter $v_i$ and costs to bribe him in the considered issue, respectively. 
The optimal solution can then be found as the entry $d_{\lvert\mathcal{B}\rvert,j}$ with maximum~$j$ 
such that $d_{\lvert\mathcal{B}\rvert,j} \leq \lfloor \frac{W}{2} \rfloor$. 
 The set of voters (which we do not bribe) corresponding to this value can be found by backtracking.

For the considered issue, we hence have a matrix with $\bigO{n\cdot nm}$ entries, each of which can be calculated in $\bigO{1}$ time. This has do be done for each issue, so the algorithm has an overall running time of $\bigO{n^2m^2}$.

With bribery actions $IV$ and $DV$, the briber is not able to bribe every voter for specific issues. Therefore we partition the set of \emph{beneficiary} voters~$\mathcal{B}$ in the set of \emph{unbribable} voters $\mathcal{B}_U$ and the set of \emph{bribable} voters $\mathcal{B}_B$. We are then searching for a set $\mathcal{S} \subseteq \mathcal{B}_B$ with $\sum_{v_i \in \mathcal{S}} \text{costs}(v_i)$ maximal such that $\sum_{v_i\in \mathcal{S}} w_i \leq \lfloor \frac{W}{2} \rfloor - \sum_{v_j \in \mathcal{B}_U} w_j$. Such a set can be obtained with the same algorithm as described above. 
\end{proof}

\subsection{Results for the unweighted negative case}

In this section, we analyze the complexity of the unweighted negative bribery problem in CP-nets. 
For the voting system $SM$ combined with cost scheme $\Cequal$, we can easily adapt the proof of the unweighted positive case and obtain \NP-completeness as well. Since we use a (parameterized) reduction from \textsc{Negative Optimal Lobbying}, this also means that the problem is \Wzwei-hard with respect to the budget, cf.\ the remark in Section~\ref{subsection:Bribery}.

\begin{theorem}\label{thm:unweighted-sm-cequal-negative}
\textsc{$(SM,A,\Cequal)$-negative-Bribery} with bribery action $A\in\{IV,DV,IV+DV\}$ is \NP-complete{}.
\end{theorem}

\begin{proof}
Theorem~\ref{thm:unweighted-sm-cequal-negative} can be shown by modifying the proof for $(SM,A,\Cequal)$-Bribery being \NP-complete~\cite[Theorem 5]{mattei2013bribery}. We just need to exchange the term `\textsc{Optimal Lobbying}' by `\textsc{Negative Optimal Lobbying}'.
\end{proof}

\begin{theorem}\label{thm:unweighted-sm-negative}
\textsc{$(SM,A,C)$-negative-Bribery} is in \Pclass{} for a bribery action $A\in\{IV,DV,IV+DV\}$ and a cost scheme $C\in \{\Cany,\Clevel,\Cflip\}$.
\end{theorem}

\begin{figure}[h!]
\centering
\begin{tikzpicture} [->,thick,align=center,xscale=2,yscale=0.6]
  \path[every node/.style={circle,minimum size=0.4cm}]
  	   (0, -1)  	node (s)  {$s$}
       (1, -2.75)	node (l1v1) {$v_3$}
       (1, -1)   	node (l1v2) {$v_2$}
       (1,0.75)	node (l1v3) {$v_1$}
       (2,-4.75) 	node (l2f1) {$X_2^\text{free}$}   
       (2,-2.25)	node (l2o1) {$X_3^5$}  
       (2,-1.25)	node (l2o2) {$X_2^5$} 
       (2,-0.25) 	node (l2o3) {$X_1^5$}     
       (2,1.25)	node (l2u1) {$X_2^4$}
       (2,2.25)	node (l2u2) {$X_1^4$}
       (3,-1) 	node (l3c1) {$v_5$}
       (3,0.5)	node (l3c2) {$v_4$}
       (4,-1) 	node (t) {$t$}
       (1,3.6) 	node (a) {$P$}
       (2,-6.3) 	node (f) {$F$}
       (3,3.6) 	node (b2) {$D$}
       (2,3.6) 	node (x) {$E$};
       
  \draw[dashed,thin,gray] (2,-4.75) ellipse (.25cm and 1.0cm);
  \draw[dashed,thin,gray] (2,0) ellipse (.3cm and 3cm);
  \path[every node/.style={font=\sffamily\small}]
    (s)		edge node[above] {} (l1v1)
    			edge node[above] {} (l1v2)
    			edge node[above] {} (l1v3)
    	(l1v1)	edge[dotted] node[above] {} (l2o1)
    			edge[dotted] node[above] {} (l2o2)
    			edge[dotted] node[above] {} (l2o3)
    			edge[dotted] node[above] {} (l2u1)
    			edge[dotted] node[above] {} (l2u2)
    			edge[dotted] node[above] {} (l2f1)	
    	(l1v2)	edge[dotted] node[above] {} (l2o1)
    			edge[dotted] node[above] {} (l2o2)
    			edge[dotted] node[above] {} (l2o3)
    			edge[dotted] node[above] {} (l2u1)
    			edge[dotted] node[above] {} (l2u2)
    			edge[dotted] node[above] {} (l2f1)
    	(l1v3)	edge[dotted] node[above] {} (l2o1)
    			edge[dotted] node[above] {} (l2o2)
    			edge[dotted] node[above] {} (l2o3)
    			edge[dotted] node[above] {} (l2u1)
    			edge[dotted] node[above] {} (l2u2)
    			edge[dotted] node[above] {} (l2f1)
    	(l2o1)	edge node[above] {} (l3c1)
    	(l2o2)	edge node[above] {} (l3c1)
    	(l2o3)	edge node[above] {} (l3c1)
    	(l2u1)	edge node[above] {} (l3c2)
    	(l2u2)	edge node[above] {} (l3c2)
    	(l2f1)	edge[dashed] node[above,pos=0.3] {$1$} (t)
    	(l3c1)	edge[dashed] node[above,pos=0.3] {$2$} (t)
    	(l3c2)	edge[dashed] node[above,pos=0.3] {$1$} (t);
   
\end{tikzpicture}
\caption{Example for the flow network of step 2 in the proof of Theorem~\ref{thm:unweighted-sm-negative} to adjust a solution for \textsc{$(SM,A,C)$-Bribery} to the negative version. In the solution to the non-negative version, the voters $v_1,v_2$, and $v_3$ are voting for $p=111$, $v_4$ for $010$, and $v_5$ for $000$. All edges have costs $0$ and capacity $1$. The only two exceptions are the dashed ones going to target $t$, and the dotted ones between $P$ and $E$. The former ones have a different capacity stated, and the latter ones have different costs. The costs of the edges $(v_i,X^j_l)$ and $(v_i,X^\text{free}_l)$ depend on whether $v_i$ was bribed in issue $X_l$.}
\label{fig:unweighted-sm-negative}
\end{figure}
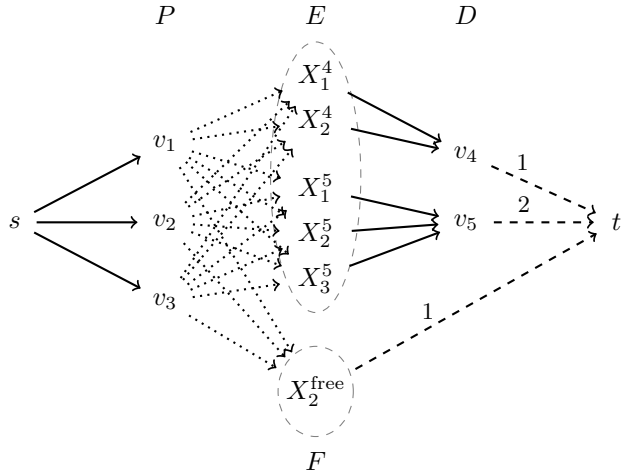

\begin{proof}
We construct the solution in two steps. First, we apply the polynomial time algorithm given by Mattei \textit{et al.}~\cite[Theorem 4]{mattei2013bribery} to solve an $(SM,A,C)$-Bribery instance. Here, a greedy strategy helps to choose which voter is to bribe in which issue. Unfortunately, the solution may contain some voters who are bribed to vote directly for $p$, therefore it does not solve \textsc{$(SM,A,C)$-negative-Bribery} directly. In the second step we \emph{repair} the first solution by choosing the cheapest flips such that no voter votes for $p$. This is done with a flow network. As the number of nodes in this network polynomially depends on the number of voters and issues, the flow problem on the network can be solved in polynomial time, too. 

\noindent An example is given in Figure~\ref{fig:unweighted-sm-negative}.

We take a source $s$ and a target $t$ for the network and construct the following sets of vertices and edges. All edges have costs $0$ and capacity $1$, unless specified otherwise.

\begin{description}
\item[P(-Voters):] For every voter $v_i$ voting for $p$ we create a node $v_i$ and an edge $(s,v_i)$.
\item[D(onor Voters):] For each voter $v_j$ not voting for $p$ we create a node $v_j$ and an edge $(v_j,t)$. Let~$d$ be the number of issues in which $v_j$'s top candidate differs from~$p$. Then the capacity of the edge $(v_j,t)$ is set to $d-1$. This ensures that $v_j$ cannot vote for $p$.
\item[F(ree Issues):] For each issue $X_l$, in which $k > \lfloor\frac{n}{2}\rfloor+1$ voters vote for $p$, we create a node $X^\text{free}_l$ and an edge $(X^\text{free}_l,t)$, with capacity $k-(\lfloor\frac{n}{2}\rfloor+1)$ . Additionally, we create an edge $(v_i, X^\text{free}_l)$ for each $v_i$ created in $P$ with costs $\costs_i(X_l)$. Flow on these edges indicates which voter has to be bribed in which issue.
\item[E(xpensive Issues):] For each node $v_j$ in $D$ we create a node $X^j_l$ if voter $v_j$ is not voting for $p$ in issue $X_l$. For each of these nodes we create an edge $(X^j_l,v_j)$. Additionally, we create all possible edges between each node $v_i$ in $P$ and each node $X^j_l$ in $E$. Those edges get different costs depending on whether $v_i$ was bribed in issue $X_l$. If so, they get the cost $\costs_j(X_l)-\costs_i(X_l)$, and $\costs_j(X_l)+\costs_i(X_l)$, otherwise. Again, flow on these edges indicates which voter has to be bribed in which issue.
\end{description}

As in the proof of Theorem~\ref{thm:OP_Cany}, one can show that there exists a flow of value~$n$ if and only if the corresponding bribery problem can be solved. The idea here is that at least one issue has to be flipped for each voter voting for $p$. This can be done easily if more than $\lfloor\frac{n}{2}\rfloor+1$ voters are voting for $p$ in this issue. Otherwise the same issue has to be flipped in the vote of a voter who is not voting for $p$. Additionally, one has to prohibit those voters to vote for $p$ in the end. This is ensured by the capacities on the edges to the target $t$. This way, the flow algorithm can find the cheapest way to \emph{repair} the solution of the first step, which leads to an optimal solution.
\end{proof}

For $OP$, we can adapt the proof of the non-negative case (Theorem~\ref{thm:OP_Cany}) and obtain the following theorem.

\begin{theorem}\label{thm:unweighted-op-negative}
\textsc{$(OP,A,C)$-negative-Bribery} is in \Pclass{} with bribery action $A\in\{IV,DV,IV+DV\}$ and a cost scheme $C\in \{\Cequal,\Cflip,\Clevel,\Cany,\Cdist\}$.
\end{theorem}

\begin{proof}
We can almost adopt the proof for Theorem~\ref{thm:OP_Cany} here, we just have to change the definition of the set of nodes~$B$: Whenever $p$ is not the top candidate of a voter $v_i$, we do not create the node $c_1^i$ and its connecting edges. Then, a voter cannot be bribed to vote for $p$. This works for $\Cany,\Clevel,\Cflip$, and $\Cdist$. For $\Cequal$, we need additional changes. Since all changes have the same costs, we can choose $n+1$ arbitrary candidates ($\neq p$) for each voter to create the nodes in $B$ together with the one for the top candidate, instead of using only the cheapest ones as before. Finally, the costs for the edges $(v_i,c^i_j)$ connecting $V$ with $B$ are set to $0$ if $c_j$ is the top candidate of $v_i$, or $1$ otherwise.
\end{proof}

\noindent To prove the next theorem, we need Lemma~$1$ from Mattei~\textit{et al.}~\cite{mattei2013bribery} once more. 

\begin{theorem}\label{thm:unweighted-ok-negative}
\textsc{$(OK^*,A,C)$-negative-Bribery} is in \Pclass{} with bribery action $A\in\{IV,DV,IV+DV\}$ and cost scheme $C\in \{\Cequal,\Cflip,\Clevel,\Cany,\Cdist\}$.
\end{theorem}

\begin{proof}
Due to Lemma~\ref{lemma:mattei1}~\cite{mattei2013bribery} we know that for $OK^*$, we only need to consider the top candidate instead of the first $k$ candidates, or, to be more precise, just the values of the first $m-j$ issues of the top candidate. So we just need to solve an instance of \textsc{$(OP,A,C)$-negative-Bribery} which is possible in polynomial time due to Theorem~\ref{thm:unweighted-op-negative}.
\end{proof}

Adapting the proof of Theorem~\ref{thm:OV_Cany}, we obtain the following theorem for the voting system~$OV$.

\begin{theorem}\label{thm:unweighted-ov-negative}
\textsc{$(OV,A,C)$-negative-Bribery} is in \Pclass{} with bribery action $A\in\{IV,DV,IV+DV\}$ and cost scheme $C\in \{\Cequal,\Cflip,\Clevel,\Cany,\Cdist\}$.
\end{theorem}

\begin{proof}\let\oldqed\qedsymbol\renewcommand{\qedsymbol}{}
We can almost adapt the proof for Theorem~\ref{thm:OV_Cany}. We once more distinguish the two cases
\begin{description}
\item[$n<2^m$:] If $p$ got at least one veto, it is a No-instance. Otherwise it is a Yes-instance.
\item[$n\geq 2^m$:] In this case we create one flow network for each $r$ with $1\leq r \leq n\cdot(2^m-1)$ as described in the proof of Theorem~\ref{thm:OV_Cany}. However, this time, we do not create the node $c_1^i$ and its incident edges if a voter $v_i$ is casting his veto against $c_1$. The rest remains unchanged. \hfill\oldqed
\end{description}
\end{proof}

\subsection{Results for the weighted negative case}

In this subsection, we derive hardness results for all variants of the weighted negative case of the bribery problem, again by reductions from the \SpecialPartition{} problem following the main idea described in Subsection~\ref{subsec:weighted-pos}.\\

\noindent
We have shown in Theorem~\ref{thm:weighted-op} in Subsection~\ref{subsec:weighted-pos} that the weighted version of the bribery problem is \NP-complete for $OP$ for all cost schemes and bribery actions $IV$ and $DV$. In the given proof, the briber is not able to transfer any points to his preferred candidate~$p$ directly, hence this also solves the negative cases for bribery actions $IV$ and $DV$. For the negative case, we can extend this result to bribery action $IV+DV$ by adding a voter $v_{\ell+2}$ with weight~$\summe$ voting for $c_4$ without dependencies. Then $p$ can only win if the briber is once more able to split the votes for $c_1$ between $c_1$ and $c_2$. We hence obtain the following corollary for Theorem~\ref{thm:weighted-op}.

\begin{corollary}\label{cor:weighted_rest}
\textsc{Weighted-$(D,A,C)$-negative-Bribery} is \NP-complete for a voting rule $D\in\{OP,OK\}$, a bribery action $A\in\{IV,DV,IV+DV\}$, and a cost scheme $C\in \{\Cequal,\Cany, \Clevel, \Cdist, \Cflip\}$.
\end{corollary}

To show \NP-completeness for the voting rule~$OV$, we can use the reduction given in the proof of Theorem~\ref{thm:ov-weighted} in Subsection~\ref{subsec:weighted-pos} for the positive version with bribery action $DV$, since in this proof, the briber never asks a voter to cast his veto from $p$ to another candidate. We will provide yet another reduction from \SpecialPartition{} to cover $IV$ and $IV+DV$ in the negative case, too. Summarizing, we obtain the following theorem.

\begin{theorem}\label{thm:weighted-OV-negative}
\textsc{Weighted-$(OV,A,C)$-negative-Bribery} with bribery action $A\in\{IV,DV,IV+DV\}$ and cost scheme $C\in \{\Cequal,\Cany, \Clevel, \Cdist, \Cflip\}$ is \NP-complete.
\end{theorem}

\begin{proof}
Theorem~\ref{thm:weighted-OV-negative} follows directly from Theorem~\ref{thm:ov-weighted} for the bribery action~$DV$. For the bribery actions~$IV$ and~$IV+DV$, we use a similar reduction as in the proof of Theorem~\ref{thm:ov-weighted}. Let $I_P$ be an instance of \SpecialPartition{} with $\mathcal{A}= \{a_1, \dots, a_{\ell}\}$ and $\sum_{a\in \mathcal{A}} a = 2\summe$. The instance $I_B$ of \textsc{Weighted-$(OV,A,C)$-negative-Bribery} that is constructed from $I_P$ has two issues~$X$ and~$Y$ with domains~$\{x,\overline{x}\}$ and ~$\{y,\overline{y}\}$, respectively, and therefore the following four candidates:

\begin{description}
\item[$xy$] The preferred candidate $p$ of the briber. There is one voter with weight~$\summe$ casting his veto against him.
\item[$\overline{x}y$] The candidate $c_2$, who starts with a weighted sum of $2\cdot\summe$ vetos.
\item[$x\overline{y}$] The candidate $c_3$, against whom no one casts a veto initially, but who will win together with the other 3 candidates, if there exists a partition.
\item[$\overline{x}\overline{y}$] An unimportant clone of $p$.
\end{description}

\noindent We achieve this by constructing the following voters:

\begin{center}
\begin{tabular}{ccccc}
\toprule 
voter~ & top candidate~ & dependency~ & weight~ & vetoes against\\
\cmidrule{1-5}
$v_1$ & $\overline{x}\overline{y}$ & none & $\summe$ & $p$  \\
$v_2$ & $xy$ & none & $\summe$ & $u$\\ 
$v_3$ & $x\overline{y}$ & $X\rightarrow Y$ & $a_1$ & $c_2$ \\ 
$v_4$ & $x\overline{y}$ & $X\rightarrow Y$ & $a_2$ & $c_2$ \\ 
\multicolumn{5}{c}{\vdots}\\ 
$v_{\ell+2}$ & $x\overline{y}$ & $X\rightarrow Y$ & $a_{\ell}$ & $c_2$ \\
\bottomrule
\end{tabular}
\end{center}

\noindent Note that due to the bribery action $IV$, the briber is only allowed to bribe the independent issue $X$ for every voter except $v_1, v_2$, so he can only transfer their vetos from $c_2$ to $c_3$. He is not allowed to bribe voter~$v_1$ in the case of negative bribery, and it is of no use to him to bribe voter~$v_2$. The best he can do is to split the vetoes of $c_2$ among $c_2$ and $c_3$, in which case they win together with $p$ and $u$. In every other case $p$ loses. With the same argument as in the proof of Theorem~\ref{thm:weighted-op}, the two instances are equivalent. Since no costs are involved here, this works for every cost scheme considered, assuming an unlimited budget.\\
The same holds for the bribery action~$IV+DV$. Here, the briber could bribe the voters $v_i,\,  i\geq 2$ to cast their veto against $c_2$, $c_3$, and $u$, but this does not help because the candidate with the fewest vetos wins and there is a total sum of $3\summe$ vetos to distribute among $u,\, c_2,\text{ and } c_3$.
\end{proof}

Last, we show that \textsc{Weighted-$(SM,A,C)$-negative-Bribery} is computationally hard, independent of the cost scheme that is used.

\begin{theorem}\label{thm:weighted-sm-negative}
\textsc{Weighted-$(SM,A,C)$-negative-Bribery} with a bribery action $A\in\{IV,DV,IV+DV\}$ and any cost scheme~$C$ is \NP-complete.
\end{theorem}

\noindent The proof uses a reduction from \SpecialPartition{} where no costs are involved. 

\begin{proof}
We prove Theorem~\ref{thm:weighted-sm-negative} by reduction from \SpecialPartition.
 Let $I_P$ be an instance of \SpecialPartition{} with $\mathcal{A}= \{a_1, \dots, a_{\ell}\}$ and $\sum_{a\in \mathcal{A}} a = 2\summe$. The instance $I_B$ of \textsc{Weighted-$(SM,A,C)$-negative-Bribery} that is constructed from $I_P$ has three issues $X,\,  Y,\,  Z$ with domains~$\{x, \overline{x}\}, \,  \{y, \overline{y}\}, \, \{z, \overline{z}\}$, respectively, and therefore eight candidates. For the bribery action $DV$, only the following four of them are important:

\begin{description}
\item[$xyz$] The preferred candidate $p$ of the briber. There is one voter voting for him weighted with $1$.
\item[$xy\overline{z}$] The candidate $c_2$, for whom no one votes initially.
\item[$x\overline{y}z$] The candidate $c_3$, for whom no one votes initially, neither. The two candidates $c_2$ and $c_3$ serve as the two partitions.
\item[$x\overline{y}\overline{z}$] The candidate $c_4$, for whom almost everyone votes initially, but who will have no voter voting for him if $p$ wins.
\end{description}\noindent We achieve this by constructing the following voters:

\begin{center}
\begin{tabular}{ccccc}
\toprule 
voter~ & top candidate~ & dependency~ & weight~ & voting for\\
\cmidrule{1-5}
$v_1$ & $xyz$ & none & $1$ & $p$  \\ 
$v_2$ & $x\overline{y}\overline{z}$ & $X\rightarrow Y,X\rightarrow Z$ & $a_1$ & $c_4$ \\ 
$v_3$ & $x\overline{y}\overline{z}$ & $X\rightarrow Y,X\rightarrow Z$ & $a_2$ & $c_4$ \\ 
\multicolumn{5}{c}{\vdots}\\ 
$v_{\ell+1}$ & $x\overline{y}\overline{z}$ & $X\rightarrow Y,X\rightarrow Z$ & $a_{\ell}$ & $c_4$ \\
\bottomrule
\end{tabular}
\end{center}

\noindent Note that the briber is only allowed to bribe the dependent issues $Y$ or $Z$ for every voter except $v_1$, so he cannot transfer points from $c_4$ to $p$. He can just split the points of $c_4$ equally among $c_2$ and $c_3$, in which case $p$ wins. In every other case $p$ loses. As in the proof of Theorem~\ref{thm:weighted-op}, one can show that the two instances are equivalent. As no costs are involved here, this works for every cost scheme considered, assuming an unlimited budget.\newline

For the case $IV$ we use a slightly different construction. In this case we do not need the issue $X$, which makes the issues $Y$ and $Z$ independent with the reduction still working. The latter version has only independent issues, so this reduction works as well for the bribery action $IV+DV$.
\end{proof}

\section{Results and discussion}\label{sec:results}
Our results are summarized in Table~\ref{table:results}.

The bribery problem and the variants of microbribery, nonuniform bribery, and swap bribery in the `classical' setting of unconditional preferences given as linear orders are tractable for the voting systems plurality and veto (\cite{faliszewski2009hard}, \cite{FHHR09}, \cite{faliszewski2008nonuniform}, \cite{elkind2009swap}). One might expect that bribing turns out to be more difficult in the case of conditional preferences and as soon as more complex cost schemes are used, but it does not. 
For the {\it non-negative unweighted} version of the bribery problem, Mattei~{\it et al.}~\cite{mattei2013bribery} obtained several tractability results (see Table~\ref{table:results}). We could solve the remaining unknown complexities for the cost schemes $\Cany$ and $\Clevel$. Contrary to the conjecture of Mattei~\textit{et al.}~\cite{mattei2013bribery}, it turned out that the bribery problem is easy in these cases as well.\\
These easiness results could be explained by the fact that $OP$ and $OV$ only require very little information on the voters' (conditional) preferences. But, more importantly, using a one-step voting rule, the bribery of one CP-net does not have an influence on other CP-nets. This is different for the sequential voting rule~$SM$. If the value of an issue is changed due to bribery of one CP-net, the cp-statements of dependent issues in other CP-nets are concerned as well, and this---in combination with the cost scheme---can make the problem potentially hard.

The interesting case might be the voting rule~$OK$: In the classical setting, the bribery problem is polynomially solvable for $k$-approval elections~\cite{faliszewski2009hard}, whereas {\sc Swap Bribery} is \NP-complete for $k\geq 2$~\cite{BD10, elkind2009swap}. So far, for bribery in CP-nets, only results for the special case of $OK$ for $\mathcal{O}$-legal profiles  and where~$k$ is a power~$j$ of 2 (denoted by $OK^*$ by Mattei \textit{et al.}) are known; it was shown by Mattei~\textit{et al.}~\cite{mattei2013bribery} that the bribery problem is solvable in polynomial time then. This
is due to the fact that in those cases, a voter always approves one \emph{package} of~$k$ candidates out of $2^{m-j}$ such packages, which are all fixed and disjoint. It is so to speak just a slightly different version of~$OP$. It would hence be interesting to investigate the computational complexity for other values of~$k$. \\

For the {\it non-negative weighted} case, we could show that finding an optimal bribery is \NP-complete for $OP$, $OV$, and $OK$ for all considered cost schemes. However, not all bribery actions are covered yet. For $OP$, $OV$ and $OK$, the computational hardness is due to the weights which enforce that a partition problem has to be solved---this is typical for the weighted variant of a problem, cf.\ the weighted versions of the original family of bribery problems in the work of  Faliszewski~\textit{et al.}~\cite{faliszewski2009hard}. However, it is interesting to see that weighted voters do not necessarily make the problem hard---this only holds in the negative case where we could show \NP-completeness for all considered variants of the problem. We have seen that the complexity for $SM$ in the non-negative case depends on the choice of the cost scheme and the cost vector, not only on the weights.\\
The complexity for $SM$ with cost scheme $\Cdist$ remains unsolved for the unweighted cases. \\

{\small
\begin{table}[htbp]
\caption{Complexity results for variants of the bribery problem in CP-nets. We distinguish solvability in polynomial time (\Pclass{}) and \NP-completeness (\NP-c). The given results all hold for the bribery actions $IV$, $DV$, and $IV+DV$, except for the results in the rows marked with~\dag, which are only shown for bribery actions $IV$ and $DV$ so far, and the ones in the row marked with~$\ddag$, which are only shown for bribery action $DV$. Results in bold face are obtained in this paper, the results in light typeface are due to Mattei~\textit{et al.}~\cite{mattei2013bribery}. $OK^*$ is the special case of $OK$ when~$k$ is a power of 2 and an $\mathcal{O}$-legal profile is given. 
The results marked with $^\lozenge$ are partly shown by Mattei~\textit{et al.}~\cite{mattei2013bribery}; they show the result only for the bribery case $IV$, in the paper on hand it is shown to hold for $IV$, $DV$ and $IV+DV$.  The cases labeled with more than one complexity class can be solved in polynomial time if the cost vector~${\bf Q}$ contains a~$1$ for each voter, and are \NP-complete for arbitrary cost vectors. In all the remaining tractable cases, the corresponding problem remains in \Pclass{} even for arbitrary cost vectors, while all of our hardness results still hold with $ Q[i]=1$ for all $1\leq i \leq n$.\label{table:results}}

\centering
\scalebox{0.86}{
\begin{tabular}{rlcccccc}
\toprule 
 & & $\Cequal$ & $\Cflip$ & $\Clevel$ & $\Cany$ & $\Cdist$ & \\
\cmidrule{2-7}
 & SM & \NP-c & \Pclass & \Pclass & \Pclass & ? & \\
 & OP & \Pclass & \Pclass & \phantom{$^\lozenge$}\POurResult{cor:op}$^\lozenge$ & \POurResult{thm:OP_Cany} & \Pclass & \texttt{Cor.\ref{cor:op}/Thm.\ref{thm:OP_Cany}}\\
 & OV & \Pclass & \Pclass & \phantom{$^\lozenge$}\POurResult{cor:ov_Clevel}$^\lozenge$ & \POurResult{thm:OV_Cany} & \Pclass & \texttt{Cor.\ref{cor:ov_Clevel}/Thm.\ref{thm:OV_Cany}}\\
 & OK* & \Pclass & \Pclass & \phantom{$^\lozenge$}\POurResult{cor:OK}$^\lozenge$ & \POurResult{cor:OK} & \Pclass & \texttt{Cor.\ref{cor:OK}}\\
\cmidrule{2-7}
weighted~ & SM & \NP-c & \Pclass{}, \NP-c & \POurResult{thm:weighted-sm-clevel}, \NP-c & \NP-c & \NP-c  & \texttt{Thm.\ref{thm:weighted-sm-clevel}}\\
 & OP\dag & \NPcOurResult{thm:weighted-op} & \NPcOurResult{thm:weighted-op} & \NPcOurResult{thm:weighted-op} & \NPcOurResult{thm:weighted-op} & \NPcOurResult{thm:weighted-op} & \texttt{Thm.\ref{thm:weighted-op}}\\
 & OV$\ddag$ & \NPcOurResult{thm:ov-weighted} & \NPcOurResult{thm:ov-weighted} & \NPcOurResult{thm:ov-weighted} & \NPcOurResult{thm:ov-weighted} & \NPcOurResult{thm:ov-weighted} & \texttt{Thm.\ref{thm:ov-weighted}}\\
 & OK\dag & \NPcOurResult{cor:weighted_OK} & \NPcOurResult{cor:weighted_OK} & \NPcOurResult{cor:weighted_OK} & \NPcOurResult{cor:weighted_OK} & \NPcOurResult{cor:weighted_OK} & \texttt{Cor.\ref{cor:weighted_OK}}\\
\cmidrule{2-7}
 negative~ & SM & \NPcOurResult{thm:unweighted-sm-cequal-negative} & \POurResult{thm:unweighted-sm-negative} & \POurResult{thm:unweighted-sm-negative} & \POurResult{thm:unweighted-sm-negative} & ? & \texttt{Thm.\ref{thm:unweighted-sm-cequal-negative}/\ref{thm:unweighted-sm-negative}}\\
 & OP & \POurResult{thm:unweighted-op-negative} & \POurResult{thm:unweighted-op-negative} & \POurResult{thm:unweighted-op-negative} & \POurResult{thm:unweighted-op-negative} & \POurResult{thm:unweighted-op-negative} & \texttt{Thm.\ref{thm:unweighted-op-negative}}\\
 & OV & \POurResult{thm:unweighted-ov-negative} & \POurResult{thm:unweighted-ov-negative} & \POurResult{thm:unweighted-ov-negative} & \POurResult{thm:unweighted-ov-negative} & \POurResult{thm:unweighted-ov-negative} & \texttt{Thm.\ref{thm:unweighted-ov-negative}}\\
 & OK* & \POurResult{thm:unweighted-ok-negative} & \POurResult{thm:unweighted-ok-negative} & \POurResult{thm:unweighted-ok-negative} & \POurResult{thm:unweighted-ok-negative} & \POurResult{thm:unweighted-ok-negative} & \texttt{Thm.\ref{thm:unweighted-ok-negative}}\\
\cmidrule{2-7}
weighted~ & SM & \NPcOurResult{thm:weighted-sm-negative} & \NPcOurResult{thm:weighted-sm-negative} & \NPcOurResult{thm:weighted-sm-negative} & \NPcOurResult{thm:weighted-sm-negative} & \NPcOurResult{thm:weighted-sm-negative} & \texttt{Thm.\ref{thm:weighted-sm-negative}}\\
negative & OP & \NPcOurResult{cor:weighted_rest} & \NPcOurResult{cor:weighted_rest} & \NPcOurResult{cor:weighted_rest} & \NPcOurResult{cor:weighted_rest} & \NPcOurResult{cor:weighted_rest} & \texttt{Cor.\ref{cor:weighted_rest}}\\
 & OV & \NPcOurResult{thm:weighted-OV-negative} & \NPcOurResult{thm:weighted-OV-negative} & \NPcOurResult{thm:weighted-OV-negative} & \NPcOurResult{thm:weighted-OV-negative} & \NPcOurResult{thm:weighted-OV-negative} & \texttt{Thm.\ref{thm:weighted-OV-negative}}\\
 & OK & \NPcOurResult{cor:weighted_rest} & \NPcOurResult{cor:weighted_rest} & \NPcOurResult{cor:weighted_rest} & \NPcOurResult{cor:weighted_rest} & \NPcOurResult{cor:weighted_rest} & \texttt{Cor.\ref{cor:weighted_rest}}\\
\bottomrule

\end{tabular}}

\end{table}
}

Summarizing, the unweighted versions do seem particularly appealing for election campaign management due to their tractability. The only exception is $SM$, depending on the cost scheme used. The weighted versions of $SM$ are all \NP-complete, with tractability for the cost schemes $\Cflip$ and $\Clevel$ if weighted $SM$ is used without individual voter costs. Of course, \NP-completeness results for the considered problems only constitute a worst case analysis and therefore cannot guarantee resistance against manipulative actions. However, they can help in acquiring a better understanding of the structure of the underlying problems and therefore may contribute in finding heuristic approaches to practically deal with them, or provide a structural decomposition for further investigations from the point of view of parameterized complexity~\cite{niedermeier2006invitation}, which again is a desired property in the setting of campaign management. \\

In respect of future research, we hope that our work contributes in creating a more extensive understanding of the nature of voting with CP-nets. The landscape of complexity in the `classical' setting where voters have unconditional preferences given as linear orders over the candidates is already quite elaborate, and it would be interesting to obtain a similar overview of the complexity of different voting problems for the CP-net setting as well.
This includes the study of additional voting rules and other common voting problems. 
Considering different voting problems for the setting of CP-nets such as the manipulation problem as initiated by Mattei~\cite{mattei2012decision} or election control will be of value for measuring vulnerability and resistance of voting in CP-nets.

Another interesting extension is the case that the dependencies of some issues are linked to those of other voters, as proposed in the setting of $m$CP-nets by Rossi~\textit{et al.}~\cite{rossi2004mcp}. In Example~\ref{ex:friends}, Bob might prefer where to go depending on Alice's choice of the destination, or even on Alice's preference when to go or what to do in the holiday.
Mattei~\textit{et al.}~\cite{mattei2013bribery} have also suggested to allow the briber to create dependencies instead of only deleting them, and to pay voters to create preferences that are conditioned by those of other voters.

\bibliographystyle{abbrv}
\bibliography{literatur}

\begin{thebibliography}{10}

\bibitem{ahuja1993network}
R.~K. Ahuja, T.~L. Magnanti, and J.~B. Orlin.
\newblock {\em Network flows - theory, algorithms and applications}.
\newblock Prentice Hall, 1993.

\bibitem{baumeister2012lazy}
D.~Baumeister, G.~Erd{\'e}lyi, E.~Hemaspaandra, L.~Hemaspaandra, and J.~Rothe.
\newblock Computational {A}spects of {A}pproval {V}oting.
\newblock In {\em Handbook of {A}pproval {V}oting}, Studies in Choice and
  Welfare, chapter~10, pages 199--251. Springer, 1 edition, 2010.

\bibitem{BD10}
N.~Betzler and B.~Dorn.
\newblock Towards a dichotomy for the {P}ossible {W}inner problem in elections
  based on scoring rules.
\newblock {\em J. Comput. Syst. Sci.}, 76(8):812--836, 2010.

\bibitem{boutilier1997constraint}
C.~Boutilier, R.~Brafman, C.~Geib, and D.~Poole.
\newblock A {C}onstraint-{B}ased {A}pproach to {P}reference {E}licitation and
  {D}ecision {M}aking.
\newblock In {\em AAAI Spring Symposium on Qualitative Decision Theory}, pages
  19--28. Citeseer, 1997.

\bibitem{boutilier2004cp}
C.~Boutilier, R.~I. Brafman, C.~Domshlak, H.~H. Hoos, and D.~Poole.
\newblock {CP}-nets: {A} {T}ool for {R}epresenting and {R}easoning with
  {C}onditional {C}eteris {P}aribus {P}reference {S}tatements.
\newblock {\em J. Artif. Intell. Res.}, 21:135--191, 2004.

\bibitem{boutilier1999reasoning}
C.~Boutilier, R.~I. Brafman, H.~H. Hoos, and D.~Poole.
\newblock Reasoning {W}ith {C}onditional {C}eteris {P}aribus {P}reference
  {S}tatements.
\newblock In {\em Proceedings of the 15th Conference on Uncertainty in
  Artificial Intelligence}, pages 71--80. Morgan Kaufmann, 1999.

\bibitem{brandt2013comsoc}
F.~Brandt, V.~Conitzer, and U.~Endriss.
\newblock Computational {S}ocial {C}hoice.
\newblock In {\em Multiagent Systems}, pages 213--283. MIT Press, 2013.

\bibitem{christian2007complexity}
R.~Christian, M.~Fellows, F.~Rosamond, and A.~Slinko.
\newblock On complexity of lobbying in multiple referenda.
\newblock {\em Rev. of Econ. Des.}, 11(3):217--224, 2007.

\bibitem{conitzer2011hypercubewise}
V.~Conitzer, J.~Lang, and L.~Xia.
\newblock Hypercubewise {P}reference {A}ggregation in {M}ulti-{I}ssue
  {D}omains.
\newblock In {\em Proceedings of the 22nd International Joint Conference on
  Artificial Intelligence}, pages 158--163. AAAI Press, 2011.

\bibitem{conitzer2015handbook}
V.~Conitzer and T.~Walsh.
\newblock Barriers to manipulation.
\newblock In {\em Handbook of Computational Social Choice}. Cambridge
  University Press, 2015.
\newblock to appear.

\bibitem{conitzer2012approximating}
V.~Conitzer and L.~Xia.
\newblock Paradoxes of {M}ultiple {E}lections: {A}n {A}pproximation {A}pproach.
\newblock In {\em Proceedings of the 13th International Conference on
  Principles of Knowledge Representation and Reasoning}, pages 179--187. AAAI
  Press, 2012.

\bibitem{dantzig1957discrete}
G.~B. Dantzig.
\newblock Discrete-{V}ariable {E}xtremum {P}roblems.
\newblock {\em Operations Res.}, 5(2):266--288, 1957.

\bibitem{dorn2012multivariate}
B.~Dorn and I.~Schlotter.
\newblock Multivariate {C}omplexity {A}nalysis of {S}wap {B}ribery.
\newblock {\em Algorithmica}, 64(1):126--151, 2012.

\bibitem{DF99}
R.~G. Downey and M.~R. Fellows.
\newblock {\em Parameterized Complexity}.
\newblock Monographs in Computer Science. Springer, New York, 1999.

\bibitem{edmonds1972theoretical}
J.~Edmonds and R.~M. Karp.
\newblock Theoretical {I}mprovements in {A}lgorithmic {E}fficiency for
  {N}etwork {F}low {P}roblems.
\newblock {\em J. of the {ACM}}, 19(2):248--264, 1972.

\bibitem{elkind2010approximation}
E.~Elkind and P.~Faliszewski.
\newblock Approximation algorithms for campaign management.
\newblock {\em Internet and Network Economics}, pages 473--482, 2010.

\bibitem{elkind2009swap}
E.~Elkind, P.~Faliszewski, and A.~Slinko.
\newblock Swap bribery.
\newblock {\em Algorithmic Game Theory}, pages 299--310, 2009.

\bibitem{eppstein2014kbest}
D.~Eppstein.
\newblock {\textdollar}k{\textdollar}-best enumeration.
\newblock {\em CoRR}, abs/1412.5075, 2014.

\bibitem{faliszewski2008nonuniform}
P.~Faliszewski.
\newblock Nonuniform {B}ribery.
\newblock In {\em Proceedings of the 7th International Joint Conference on
  Autonomous Agents and Multiagent Systems}, volume~3, pages 1569--1572.
  {IFAAMAS}, 2008.

\bibitem{faliszewski2006complexity}
P.~Faliszewski, E.~Hemaspaandra, and L.~A. Hemaspaandra.
\newblock The {C}omplexity of {B}ribery in {E}lections.
\newblock In {\em Proceedings of the 21st National Conference on Artificial
  Intelligence}, volume~6, pages 641--646. {AAAI} Press, 2006.

\bibitem{faliszewski2009hard}
P.~Faliszewski, E.~Hemaspaandra, and L.~A. Hemaspaandra.
\newblock How {H}ard {I}s {B}ribery in {E}lections?
\newblock {\em J. Artif. Intell. Res.}, 35(2):485--532, 2009.

\bibitem{faliszewski2010protect}
P.~Faliszewski, E.~Hemaspaandra, and L.~A. Hemaspaandra.
\newblock Using {C}omplexity to {P}rotect {E}lections.
\newblock {\em Commun. {ACM}}, 53(11):74--82, 2010.

\bibitem{FHHR09}
P.~Faliszewski, E.~Hemaspaandra, L.~A. Hemaspaandra, and J.~Rothe.
\newblock Llull and {C}opeland {V}oting {C}omputationally {R}esist {B}ribery
  and {C}onstructive {C}ontrol.
\newblock {\em J. Artif. Intell. Res.}, 35(1):275--341, 2009.

\bibitem{faliszewski2008copelandties}
P.~Faliszewski, E.~Hemaspaandra, and H.~Schnoor.
\newblock Copeland voting: {T}ies matter.
\newblock In {\em Proceeding of the 7th International Joint Conference on
  Autonomous Agents and Multiagent Systems}, volume~2, pages 983--990. IFAAMAS,
  2008.

\bibitem{faliszewski2009richer}
P.~Faliszewski, L.~Hemaspaandra, E.~Hemaspaandra, , and J.~Rothe.
\newblock A {R}icher {U}nderstanding of the {C}omplexity of {E}lection
  {S}ystems.
\newblock In {\em Fundamental Problems in Computing: Essays in Honor of
  Professor Daniel J. Rosenkrantz}, chapter~14, pages 375--406. Springer, 1
  edition, 2009.

\bibitem{faliszewski2010war}
P.~Faliszewski and A.~D. Procaccia.
\newblock {AI}'s {W}ar on {M}anipulation: {A}re {W}e {W}inning?
\newblock {\em {AI} Magazine}, 31(4):53--64, 2010.

\bibitem{faliszewski2015handbook}
P.~Faliszewski and J.~Rothe.
\newblock Control and bribery.
\newblock In {\em Handbook of Computational Social Choice}. Cambridge
  University Press, 2015.
\newblock to appear.

\bibitem{FG06}
J.~Flum and M.~Grohe.
\newblock {\em Parameterized Complexity Theory}.
\newblock Texts in Theoretical Computer Science. An EATCS Series. Springer, New
  York, 2006.

\bibitem{ford1956maximal}
L.~R. Ford and D.~R. Fulkerson.
\newblock Maximal flow through a network.
\newblock {\em Can. J. of Math.}, 8(3):399--404, 1956.

\bibitem{gary1979computers}
M.~Garey and D.~Johnson.
\newblock {\em {Computers and Intractability: A Guide to the Theory of
  NP-Completeness}}.
\newblock WH Freeman and Company, New York, 1979.

\bibitem{karp1972reducibility}
R.~M. Karp.
\newblock Reducibility among combinatorial problems.
\newblock In {\em Complexity of Computer Computations}, The IBM Research
  Symposia Series, pages 85--103. Springer, 1972.

\bibitem{Knuth1998theart3}
D.~E. Knuth.
\newblock {\em The Art of Computer Programming, Volume 3: (2nd Ed.) Sorting and
  Searching}.
\newblock Addison Wesley Longman Publishing Co., Inc., Redwood City, CA, USA,
  1998.

\bibitem{lang2007vote}
J.~Lang.
\newblock Vote and {A}ggregation in {C}ombinatorial {D}omains with {S}tructured
  {P}references.
\newblock In {\em Proceedings of the 20th International Joint Conference on
  Artificial Intelligence}, pages 1366--1371, 2007.

\bibitem{lang2009sequential}
J.~Lang and L.~Xia.
\newblock Sequential composition of voting rules in multi-issue domains.
\newblock {\em Math. Soc. Sci.}, 57(3):304--324, 2009.

\bibitem{lawler1972procedure}
E.~L. Lawler.
\newblock {A Procedure for Computing the K Best Solutions to Discrete
  Optimization Problems and Its Application to the Shortest Path Problem}.
\newblock {\em Management Science}, 18(7):401--405, 1972.

\bibitem{li2010efficient}
M.~Li, Q.~B. Vo, and R.~Kowalczyk.
\newblock An {E}fficient {P}rocedure for {C}ollective {D}ecision-making with
  {CP}-nets.
\newblock In {\em Proceedings of the 19th European Conference on Artificial
  Intelligence}, volume 215 of {\em FAIA}, pages 375--380. {IOS} Press, 2010.

\bibitem{li2011majority}
M.~Li, Q.~B. Vo, and R.~Kowalczyk.
\newblock Majority-rule-based preference aggregation on multi-attribute domains
  with {CP}-nets.
\newblock In {\em Proceedings of the 10th International Conference on
  Autonomous Agents and Multiagent Systems}, volume~2, pages 659--666.
  {IFAAMAS}, 2011.

\bibitem{maran2013framework}
A.~Maran, N.~Maudet, M.~S. Pini, F.~Rossi, and K.~B. Venable.
\newblock A {F}ramework for {A}ggregating {I}nfluenced {CP}-{N}ets and {I}ts
  {R}esistance to {B}ribery.
\newblock In {\em Proceedings of the 27th {AAAI} Conference on Artificial
  Intelligence}, pages 668--674. {AAAI} Press, 2013.

\bibitem{mattei2012decision}
N.~Mattei.
\newblock {\em Decision {M}aking {U}nder {U}ncertainty: {T}heoretical and
  {E}mpirical {R}esults on {S}ocial {C}hoice, {M}anipulation, and {B}ribery.}
\newblock PhD thesis, University of Kentucky, 2012.

\bibitem{DBLP:journals/corr/abs-1304-6174}
N.~Mattei, N.~Narodytska, and T.~Walsh.
\newblock How hard is it to control an election by breaking ties?
\newblock In {\em Proceedings of the 21st European Conference on Artificial
  Intelligence}, volume 263 of {\em FAIA}, pages 1067--1068. {IOS} Press, 2014.

\bibitem{mattei2013bribery}
N.~Mattei, M.~S. Pini, F.~Rossi, and K.~B. Venable.
\newblock Bribery in voting with {CP}-nets.
\newblock {\em Ann. Math. Artif. Intell.}, 68(1-3):135--160, 2013.

\bibitem{niedermeier2006invitation}
R.~Niedermeier.
\newblock {\em Invitation to {F}ixed-{P}arameter {A}lgorithms}.
\newblock Number 31 in Oxford Lecture Series in Mathematics and Its
  Applications. Oxford University Press, 2006.

\bibitem{Obraztsova:2011:CVM:2283396.2283450}
S.~Obraztsova and E.~Elkind.
\newblock On the {C}omplexity of {V}oting {M}anipulation {U}nder {R}andomized
  {T}ie-breaking.
\newblock In {\em Proceedings of the 22nd International Joint Conference on
  Artificial Intelligence}, volume~1, pages 319--324. AAAI Press, 2011.

\bibitem{obraztsova2011ties}
S.~Obraztsova, E.~Elkind, and N.~Hazon.
\newblock Ties matter: {C}omplexity of voting manipulation revisited.
\newblock In {\em Proceedings of the 10th International Conference on
  Autonomous Agents and Multiagent Systems}, volume~1, pages 71--78. IFAAMAS,
  2011.

\bibitem{pini2013bribery}
M.~Pini, F.~Rossi, and K.~Venable.
\newblock Bribery in {V}oting {W}ith {S}oft {C}onstraints.
\newblock In {\em Proceedings of the 27th {AAAI} Conference on Artificial
  Intelligence}, pages 803--809. AAAI Press, 2013.

\bibitem{purrington2007making}
K.~Purrington and E.~H. Durfee.
\newblock {M}aking {S}ocial {C}hoices from {I}ndividuals' {CP}-nets.
\newblock In {\em Proceedings of the 6th International Joint Conference on
  Autonomous Agents and Multiagent Systems}, pages 1122--1124. {IFAAMAS}, 2007.

\bibitem{rossi2004mcp}
F.~Rossi, K.~B. Venable, and T.~Walsh.
\newblock m{CP} {N}ets: {R}epresenting and {R}easoning with {P}references of
  {M}ultiple {A}gents.
\newblock In {\em Proceedings of the 19th National Conference on Artificial
  Intelligence}, pages 729--734. {AAAI} Press, 2004.

\bibitem{rothe2013shields}
J.~Rothe and L.~Schend.
\newblock Challenges to complexity shields that are supposed to protect
  elections against manipulation and control: a survey.
\newblock {\em Ann. Math. Artif. Intell.}, 68(1-3):161--193, 2013.

\bibitem{schlotter2011campaign}
I.~Schlotter, E.~Elkind, and P.~Faliszewski.
\newblock Campaign {M}anagement under {A}pproval-{D}riven {V}oting {R}ules.
\newblock In {\em Proceedings of the 25th AAAI Conference on Artificial
  Intelligence}, pages 726--731, 2011.

\bibitem{xia2008voting}
L.~Xia, V.~Conitzer, and J.~Lang.
\newblock Voting on {M}ultiattribute {D}omains with {C}yclic {P}referential
  {D}ependencies.
\newblock In {\em Proceedings of the 23rd {AAAI} Conference on Artificial
  Intelligence}, pages 202--207. {AAAI} Press, 2008.

\bibitem{xia2007sequential}
L.~Xia, J.~Lang, and M.~Ying.
\newblock Sequential voting rules and multiple elections paradoxes.
\newblock In {\em Proceedings of the 11th Conference on Theoretical Aspects of
  Rationality and Knowledge}, pages 279--288. ACM, 2007.

\bibitem{xia2007strongly}
L.~Xia, J.~Lang, and M.~Ying.
\newblock Strongly {D}ecomposable {V}oting {R}ules on {M}ultiattribute
  {D}omains.
\newblock In {\em Proceedings of the 22nd {AAAI} Conference on Artificial
  Intelligence}, pages 776--781. {AAAI} Press, 2007.

\end{thebibliography}
\end{document}